\documentclass[envcountsame]{llncs}
\usepackage[colorlinks=true, citecolor=blue,linkcolor=red,urlcolor=black]{hyperref}
\usepackage[utf8]{inputenc}
\usepackage{fullpage}
\usepackage{hhline}

\usepackage{xcolor,tikz}
\usetikzlibrary{automata,arrows,shapes,backgrounds,decorations.pathreplacing}

\usepackage{amsmath}
\usepackage{amssymb,stmaryrd}
\usepackage{enumitem}
\usepackage{multirow}

\title{Window parity games: an alternative approach toward parity games with time bounds\thanks{Q.~Hautem is supported by a FRIA fellowship, M.~Randour is an F.R.S.-FNRS Postdoctoral researcher.}}

\author{V\'eronique Bruy\`ere\inst{1} \and Quentin Hautem\inst{1} \and Mickael Randour\inst{1,2}}	 
\institute{
Computer Science Department, Université de Mons (UMONS), Belgium\\
\and Computer Science Department, Universit\'e libre de Bruxelles (ULB), Belgium}

\newcommand{\N}{\mathbb{N}}

\newcommand{\Nzero}{\N_0}

\newcommand{\ssetminus}{\! \setminus \!}

\newcommand{\stronger}{\preceq}
\newcommand{\strongerm}{\preceq_m}
\newcommand{\smallest}{$\stronger$-smallest}
\newcommand{\smaller}{$\stronger$-smaller}

\newcommand{\playerOne}{\ensuremath{\mathcal{P}_1} } 
\newcommand{\playerTwo}{\ensuremath{\mathcal{P}_2 } }
\newcommand{\playerI}{\ensuremath{\mathcal{P}_i} }
\newcommand{\Gdev}{(V_1,V_2, E)}   
\newcommand{\Plays}{\mathsf{Plays}}

\newcommand{\Out}{\mathsf{Out}}
\newcommand{\Obj}{\Omega}  
\newcommand{\WinG}[3]{{\mathsf{Win}}_{#1}^{#3}({#2})}  

\newcommand{\pifactor}[2]{\pi{[{#1},{#2}]}}

\newcommand{\Reach}{\mathsf{Reach}}
\newcommand{\Safe}{\mathsf{Safe}}
\newcommand{\Buchi}{\mathsf{Buchi}}
\newcommand{\CoBuchi}{\mathsf{CoBuchi}}
\newcommand{\GenReach}{\mathsf{GenReach}}

\newcommand{\Par}{\mathsf{Parity}}

\newcommand{\Inf}{\mathsf{Inf}}

\newcommand{\DirFWP}{\mathsf{DirFixWP}}
\newcommand{\FWP}{\mathsf{FixWP}}
\newcommand{\DirBWP}{\mathsf{DirBndWP}}
\newcommand{\BWP}{\mathsf{BndWP}}
\newcommand{\DirBP}{\mathsf{DirBndPR}}
\newcommand{\BP}{\mathsf{BndPR}}
\newcommand{\DirPR}{\mathsf{DirFixPR}}
\newcommand{\PR}{\mathsf{FixPR}}
\newcommand{\RR}{\mathsf{RR}}

\newcommand{\BndRR}{\mathsf{BndRR}}
\newcommand{\DirBndRR}{\mathsf{DirBndRR}}

\newcommand{\window}[1]{window at position~${#1}$}
\newcommand{\closed}{closed}

\newcommand{\ffclosed}[1]{first-closed}
\newcommand{\good}{$\lambda$-good}
\newcommand{\bad}{bad}

\newcommand{\GoodDec}[1]{${#1}$-good decomposition}
\newcommand{\EGoodDec}[1]{eventually ${#1}$-good decomposition}

\begin{document}
  \maketitle	

\begin{abstract} 
Classical objectives in two-player zero-sum games played on graphs often deal with limit behaviors of infinite plays: e.g., \textit{mean-payoff} and \textit{total-payoff} in the quantitative setting, or \textit{parity} in the qualitative one (a canonical way to encode $\omega$-regular properties). Those objectives offer powerful abstraction mechanisms and often yield nice properties such as memoryless determinacy.
However, their very nature provides no guarantee on time bounds within which something good can be witnessed. 
In this work, we consider two approaches toward inclusion of \textit{time bounds} in parity games. The first one, \textit{parity-response} games, is based on the notion of finitary parity games~\cite{ChatterjeeHH09} and parity games with costs~\cite{DBLP:journals/corr/abs-1207-0663,Weinert016}. The second one, \textit{window parity} games, is inspired by window mean-payoff games~\cite{Chatterjee0RR15}. We compare the two approaches and show that while they prove to be equivalent in some contexts, window parity games offer a more tractable alternative when the time bound is given as a parameter ($\sf P$-c.~vs. $\sf PSPACE$-c.). In particular, it provides a conservative approximation of parity games computable in polynomial time. Furthermore, we extend both approaches to the multi-dimension setting. We give the full picture for both types of games with regard to complexity and memory bounds.
\end{abstract}

\section{Introduction}
\label{sec:intro}

\paragraph{\bf Games on graphs.} Two-player games played on directed graphs constitute an important framework for the synthesis of a suitable controller for a reactive system faced to an uncontrollable environment~\cite{randourECCS}. In this setting, \textit{vertices} of the graph represent states of the system and edges represent transitions between those states. We consider \textit{turn-based two-player} games: each vertex either belongs to the system (the first player, denoted by~$\playerOne$) or the environment (the second player, denoted by $\playerTwo$). A game is played by moving an imaginary pebble from vertex to vertex according to existing transitions: the owner of a vertex decides where to move the pebble. The outcome of the game is an infinite sequence of vertices called \textit{play}. The choices of both players depend on their respective \textit{strategy} which can use an arbitrary amount of memory in full generality. In the classical setting, $\playerOne$ tries to achieve an \textit{objective} (describing a set of \textit{winning plays}) while $\playerTwo$ tries to prevent him from succeeding: hence, our games are \textit{zero-sum}. As all the objectives considered in this paper define Borel sets, Martin's theorem~\cite{Martin75} guarantees determinacy.

\paragraph{\bf Parity games.} Two-player games with $\omega$-regular objectives have been studied extensively in the literature. See for example~\cite{Thomas97,2001automata} for an introduction. A canonical way to represent games with $\omega$-regular conditions is the class of \textit{parity games}: vertices are assigned a non-negative integer \textit{priority} (or \textit{color}), and the objective asks that among the vertices that are seen infinitely often along a play, the minimal priority be even. Parity games have been under close scrutiny for a long time both due to their importance (e.g., they subsume modal $\mu$-calculus model checking~\cite{DBLP:conf/cav/EmersonJS93}) and their intriguing complexity: they belong to the class of problems in $\sf UP \cap coUP$~\cite{Jurdzinski98} and despite many efforts (e.g.,~\cite{Zielonka98,Jurdzinski00,DBLP:journals/siamcomp/JurdzinskiPZ08,DBLP:conf/fsttcs/Schewe07}), whether they belong to $\sf P$ is still an open question. Furthermore, parity games enjoy memoryless determinacy~\cite{EmersonJ88,Zielonka98}. Multi-dimension parity games were studied in~\cite{DBLP:conf/fossacs/ChatterjeeHP07}: in such games, $n$-dimension vectors of priorities are associated to each vertex, and the objective is to satisfy the conjunction of all the one-dimension parity objectives. The complexity of solving those games is higher: deciding if $\playerOne$ (resp.~$\playerTwo$) has a winning strategy is $\sf coNP$-complete (resp.~$\sf NP$-complete) and exponential memory is needed for $\playerOne$ whereas $\playerTwo$ remains memoryless~\cite{DziembowskiJW97,Hor05-GDV}

\paragraph{\bf Time bounds.} In its classical formulation, the parity objective essentially requires that \textit{for each odd priority seen infinitely often, a smaller even priority should also be seen infinitely often}. An odd priority can be seen as a stimulus that must be answered by seeing a smaller even priority. The parity objective has fundamental qualities. The simplicity of its definition totally abstracts timing issues like ``how much time has elapsed between a stimulus and its answer'' and is key to memoryless determinacy. This makes it robust to slight changes in the model which could impact more precise formulations (e.g., counting the number of steps between a stimulus and its answer critically depends on the granularity of the game graph).

Nonetheless, it has been recently argued that in a large number of practical applications, \textit{timing does matter} (e.g.,~\cite{ChatterjeeHH09,DBLP:journals/fmsd/KupfermanPV09,Chatterjee0RR15}). Indeed, in general it does not suffice to know that a ``good behavior'' will \textit{eventually} happen, and one wants to ensure that it can actually be witnessed \textit{within a time frame which is acceptable} with regard to the modeled reactive system. For example, consider a computer server having to grant requests to clients. A classical parity objective can encode that requests should eventually be granted. However, it is clear that in a desired controller, requests should not be placed on hold for an arbitrarily long time. In order to accomodate such requirements, various attempts to associate classical game objectives with time bounds have been recently studied. For example, \textit{window mean-payoff} and \textit{window total-payoff} games provide a framework to reason about quantitative games (e.g., modeling quantities such as energy consumption) with time bounds~\cite{Chatterjee0RR15}. In the qualitative setting, \textit{finitary parity} games~\cite{ChatterjeeHH09,ChatterjeeF13} and \textit{parity games with costs}~\cite{DBLP:journals/corr/abs-1207-0663,Weinert016} provide a similar framework for parity games.

\paragraph{\bf Two approaches.} While window games and finitary parity games (resp.~parity games with costs) share the goal of allowing precise specification of time bounds, their inner mechanisms differ. The aim of our work is three-fold: $(i)$~apply the \textbf{window mechanism to parity games}, $(ii)$~provide a \textbf{thorough comparison} with the existing framework of finitary parity games and parity games with costs, $(iii)$~\textbf{extend both approaches to the multi-dimension setting} (which was left unexplored up to now). Since all
those related papers do not use a uniform terminology, we here use the following taxonomy for the two approaches.
\begin{itemize}
\item \textbf{Window parity (WP).} Intuitively, the \textit{direct fixed WP} objective considers a window of size bounded by $\lambda \in \mathbb{N}_0$ (given as a parameter) sliding over an infinite play and declare this play winning if in all positions, the window is such that the minimal priority within it is even. For \textit{direct bounded WP}, the size of the window is not fixed as a parameter but a play is winning if there exists a bound $\lambda$ for which the condition holds. We also consider the \textit{fixed WP} and \textit{bounded WP} objectives which are essentially prefix-independent variants of the previous ones. All those objectives are based on the window mechanism introduced in~\cite{Chatterjee0RR15} and our work presents the first implementation of this mechanism for parity games.

\item \textbf{Parity-response (PR).} The \textit{direct fixed PR} objective asks that along a play, any odd priority be followed by a smaller even priority in at most $\lambda \in \mathbb{N}_0$ (given as a parameter) steps. As for the WP setting, we also consider the \textit{direct bounded PR} objective where a play is winning if there exists a bound~$\lambda$ such that the condition holds, along with the respective prefix-independent variants: the \textit{fixed PR} and the \textit{bounded PR} objectives. The \textit{bounded PR} objective was studied for one-dimension games (under the name \textit{finitary parity}) in~\cite{ChatterjeeHH09}: deciding the winner is in $\sf P$ and memoryless strategies suffice for $\playerOne$ while $\playerTwo$ needs infinite memory. The \textit{fixed PR} objective for one-dimension games was very recently proved to be $\sf PSPACE$-complete, with exponential memory bounds for both players~\cite{Weinert016} (this work is presented in the more general context of \textit{parity games with costs}). Our work provides the first study of the parity-response approach in multi-dimension games.
\end{itemize}

\renewcommand{\arraystretch}{1.4}
\begin{table}
\begin{center}
\begin{tabular}{|c||c|c|c||c|c|c|}
\cline{2-7}
\multicolumn{1}{c|}{} & \multicolumn{3}{c||}{one-dimension} & \multicolumn{3}{c|}{ multi-dimension} \\
\cline{2-7}
\multicolumn{1}{c|}{} & ~complexity~ & $\playerOne$ mem. & $\playerTwo$ mem. & ~complexity~ & $\playerOne$ mem. & $~\playerTwo$ mem.$~$ \\
\hhline{-|======|}
\textbf{Fixed WP} & \textbf{\textsf{P}-c.} & \multicolumn{2}{c||}{\textbf{polynomial}}  & ~\multirow{4}{*}{\textbf{\textsf{EXPTIME}-c.}}~ & \multicolumn{2}{c|}{\multirow{2}{*}{\textbf{exponential}}} \\
\cline{1-4}
Fixed PR & \textsf{PSPACE}-c. & ~exponential~ & ~$\leq$ exponential~ & & \multicolumn{2}{c|}{} \\
\cline{1-4} \cline{6-7}
~\textbf{Bounded WP}~ & \multirow{2}{*}{\textsf{P}-c.} & \multirow{2}{*}{memoryless} & \multirow{2}{*}{infinite} & & \multirow{2}{*}{$~$\textbf{exponential}$~$} & \multirow{2}{*}{\textbf{infinite}} \\ 
\cline{1-1}
~Bounded PR~ & & & & & & \\
\cline{1-7}
\end{tabular}
\end{center}
\caption{Complexity of deciding the winner and memory required for winning strategies in window parity (WP) and parity-response (PR) games. All results hold for both the \textit{prefix-independent} and the \textit{direct} ({\sf Dir}) variants of all the objectives, except for the memory of $\playerTwo$ in the direct bounded cases: in one-dimension games, linear memory is both sufficient and necessary (for both WP and PR) and in multi-dimension games, exponential memory is both sufficient and necessary. All bounds are tight unless specified by the $\leq$ symbol. New results are in bold.}
\label{table:overview}
\end{table}

\paragraph{\bf Our contributions.} Given the number of variants studied, we give an overview of our results in Table~\ref{table:overview}. Our main contributions are as follows. 
\begin{enumerate}
\item We prove that \textit{bounded WP} and \textit{bounded PR} objectives coincide, even in multi-dimension games (Proposition~\ref{prop:inclusions}).
\item We establish that \textit{bounded WP} (and thus \textit{bounded PR}) games are $\sf P$-hard in one-dimension (Theorem~\ref{thm:finitaryparity}, $\sf P$-membership follows from~\cite{ChatterjeeHH09}) and that they are $\sf EXPTIME$-complete in multi-dimension (Theorem~\ref{prop:multiDirBP}). The $\sf EXPTIME$-membership follows from a reduction to a variant of \textit{request-response games}~\cite{WallmeierHT03} presented in~\cite{ChatterjeeHH09} under the name of \textit{finitary Streett games} (Lemma~\ref{lem:finitarystreett}). The $\sf EXPTIME$-hardness is proved via a reduction from the membership problem in alternating polynomial-space Turing machines (Lemma~\ref{prop:exptimehard}).
\item We show that in multi-dimension \textit{bounded WP} (and thus \textit{bounded PR}) games, exponential memory is both sufficient and necessary for $\playerOne$ while infinite memory is needed for $\playerTwo$ (Theorem~\ref{prop:multiDirBP}). 
\item We prove that one-dimension \textit{fixed WP} games provide a \textbf{conservative approximation of parity games} (Proposition~\ref{prop:inclusions}) \textbf{computable in polynomial time} (Theorem~\ref{prop:WPonedim}). This is in contrast to the $\sf PSPACE$-completeness of \textit{fixed PR} games~\cite{Weinert016} (actually, the proof in~\cite{Weinert016} is for a more general model but already holds for \textit{fixed PR} games). 
\item While \textit{fixed PR} games are $\sf PSPACE$-complete, we establish two polynomial-time algorithms (Theorem~\ref{thm:fixed}) to solve fixed-parameter sub-cases: $(i)$ the bound $\lambda$ is fixed, or $(ii)$ the number of priorities is fixed.
\item In multi-dimension, we prove that both \textit{fixed PR} (Theorem~\ref{prop:multiPR}) and \textit{fixed WP} (Theorem~\ref{prop:multiWP}) games are $\sf EXPTIME$-complete. Membership relies on different techniques and algorithms for each case while hardness is based on the same reduction as for the \textit{bounded} variants (Lemma~\ref{prop:exptimehard}).
\item In one-dimension games, we also establish that for \textit{fixed WP}, polynomial memory is both sufficient and necessary for both players, whereas exponential memory is required for \textit{fixed PR}~\cite{Weinert016}. In multi-dimension games, we prove that for both \textit{fixed PR} and \textit{fixed WP}, exponential memory is both sufficient and necessary for both players. The upper bounds follow from the $\sf EXPTIME$ algorithms mentioned above whereas the lower bounds in one-dimension are shown thanks to appropriate families of games (Example~\ref{ex:necmemory}) and in multi-dimension are obtained through reduction from \textit{generalized reachability games}~\cite{FijalkowH13} (Lemma~\ref{lem:memexp}).
\item We establish the existence of values of $\lambda$ such that the \textit{fixed} objectives become equivalent to the \textit{bounded} ones, both in one-dimension (Theorem~\ref{thm:finitaryparity}) and in multi-dimension (Theorem~\ref{prop:multiDirBP}).
\end{enumerate}
While all the aforementioned results are for the \textit{prefix-independent} variants of our objectives, we also obtain closely related complexities and memory bounds for the \textit{direct} ones (Table~\ref{table:overview}). We obtain our results using a variety of techniques, sometimes inspired by~\cite{ChatterjeeHH09,Chatterjee0RR15}. Our focus is on giving the full picture for the two approaches toward including time bounds in parity games: window parity and parity-response. We sum up the key comparison points in the next paragraph.

\paragraph{\bf Comparison.} The \textit{parity-response} and \textit{window parity} approaches turn out to be equivalent in the \textit{bounded} context, i.e., when the question is the existence of a bound $\lambda \in \mathbb{N}_0$ for which the corresponding \textit{fixed} objective holds. Hence, the focus of the comparison is the \textit{fixed} variants. Observe that those variants are of interest for applications where the time bound is part of the specification: parameter $\lambda$ grants flexibility in the specification as it can be adjusted to specific requirements of the application. Let us review the complexities of the \textit{fixed PR} and \textit{fixed WP} objectives.

In one-dimension games, \textit{fixed PR} is $\sf PSPACE$-complete whereas \textit{\textit{fixed WP} provides a framework with similar flavor that enjoys increased tractability}: it is $\sf P$-complete. Hence, \textit{\textit{fixed WP} does provide a polynomial-time conservative approximation of parity games} (Proposition~\ref{prop:inclusions}). Interestingly, the \textit{fixed WP} objective also permits to approximate the \textit{fixed PR} one in both directions, and in polynomial time: we prove in Proposition~\ref{prop:inclusions} that the \textit{fixed PR} objective for time bound $\lambda$ can be framed by the \textit{fixed WP} objective for two well-chosen values of the time bound $\lambda'$ and $\lambda''$.

In multi-dimension, both \textit{fixed PR} and \textit{fixed WP} games are $\sf EXPTIME$-complete. Nonetheless, while the algorithm for \textit{fixed PR} requires exponential time in \textit{both} the number of dimensions \textit{and} the number of priorities (which can be as large as the game graph), \textit{solving the \textit{fixed WP} case only requires exponential time in the number of dimensions}. This distinction may have \textit{impact on practical applications} where, usually, the size of the model (hence the game graph) can be very large while the specification (hence the number of dimensions) is comparatively small. Note that for both objectives, the multi-dimension algorithms are \textit{pseudo-}polynomial in the time bound $\lambda$, hence also exponential in the length of its binary encoding.

Finally, let us compare \textit{window parity} games with \textit{window mean-payoff (WMP)} games~\cite{Chatterjee0RR15}. First, one could naturally wonder if WP games could be solved by encoding them into WMP games, following a reduction similar in spirit to the one developed by Jurdzinski for classical parity games~\cite{Jurdzinski98}. This is indeed possible, but leads to increased complexities in comparison to the ad hoc analysis developed in this work. For example, multi-dimension \textit{fixed WP} games would require exponential time in the number of priorities too. Second, observe that \textit{fixed WP} games can be solved in polynomial time whatever the bound $\lambda \in \mathbb{N}_0$ whereas \textit{fixed WMP} games require pseudo-polynomial time, i.e., also polynomial in the bound $\lambda$. Finally, multi-dimension \textit{bounded WMP} games are known to be non-primitive-recursive-hard and their decidability is still open~\cite{Chatterjee0RR15}. On the contrary, multi-dimension \textit{bounded WP} games are $\sf EXPTIME$-complete. This suggests that the colossal complexity of \textit{bounded WMP} games is a result of the quantitative nature of mean-payoff mixed with windows, and not an inherent drawback of the window mechanism.

\paragraph{\bf Other related work.} This paper extends its preceding conference version~\cite{BHR16}. In addition to the aforementioned articles, we mention two papers where logical formalisms dealing with time bounds are studied. In~\cite{DBLP:journals/fmsd/KupfermanPV09}, Kupferman et al.~introduced Prompt-LTL, which is strongly linked with the finitary conditions discussed above. In~\cite{DBLP:conf/csl/BaierKKW14}, Baier et al.~also studied an extension of LTL that can express properties based on the window mechanism of~\cite{Chatterjee0RR15}. The study of logical fragments corresponding to our framework of window parity games is an interesting question left open for future work.

\paragraph{\bf Outline.} Section~\ref{sec:preliminaries} presents the needed definitions and known results about classical objectives. Section~\ref{sec:objectives} introduces the different objectives studied in this paper and establishes the links between them. Section~\ref{sec:onedim} and Section~\ref{sec:manydim} respectively present our results for one-dimension and multi-dimension games. 
\section{Preliminaries} \label{sec:preliminaries}

\paragraph{\bf Game structures.}
We consider zero-sum turn-based games played by two players, $\playerOne$ and $\playerTwo$, on a finite directed graph.

\begin{definition}
A \emph{game structure} is a tuple $G = (V_1,V_2,E)$ where
\begin{itemize}
\item $(V,E)$ is a finite directed graph, with $V = V_1 \cup V_2$ the set of vertices and $E \subseteq V \times V$ the set of edges such that for each $v \in V$, there exists $(v,v') \in E$ for some $v' \in V$ (no deadlock),
\item $(V_1,V_2)$ forms a partition of $V$ such that $V_i$ is the set of vertices controlled by player $\playerI$ with $i \in \{1,2\}$.
\end{itemize}
\end{definition}

A \emph{play} of $G$ is an infinite sequence of vertices $\pi = v_0 v_1 \ldots \in V^{\omega}$ such that $(v_k,v_{k+1}) \in E$ for all $k \in \N$. We denote by $\Plays(G)$ the set of plays in $G$. \emph{Histories} of $G$ are finite sequences $\rho = v_0 \ldots v_k \in V^\ast$ defined in the same way. Given a play $\pi = v_0 v_1 \ldots$, the history $v_k \ldots v_{k+l}$ is denoted by $\pifactor{k}{k+l}$; in particular, $v_k = \pi[k]$. We also use notation $\pifactor{k}{\infty}$ for the suffix $v_k v_{k+1} \ldots$ of $\pi$.

\paragraph{\bf Strategies.} A \emph{strategy} $\sigma_i$ for $\playerI$ is a function $\sigma_i\colon V^*V_i \rightarrow V$ assigning to each history $\rho v \in V^*V_i$ a vertex $v' = \sigma_i(\rho v)$ such that $(v,v') \in E$. It is \emph{memoryless} if $\sigma_i(\rho v) = \sigma_i(\rho'v)$ for all histories $\rho v, \rho'v$ ending with the same vertex $v$, that is, if $\sigma_i$ is a function $\sigma_i\colon V_i \rightarrow V$. It is \emph{finite-memory} if it can be encoded by a deterministic \emph{Moore machine} $(M, m_0, \alpha_u, \alpha_n)$ where $M$ is a finite set of states (the memory of the strategy), $m_0 \in M$ is the initial memory state, $\alpha_u\colon M \times V \rightarrow M$ is the update function, and $\alpha_n\colon M \times V_i \rightarrow V$ is the next-action function. The Moore machine defines a strategy $\sigma_i$ such that $\sigma_i(\rho v) = \alpha_n(\widehat{\alpha}_u(m_0,\rho),v)$ for all histories $\rho v \in V^*V_i$, where $\widehat{\alpha}_u$ extends $\alpha_u$ to sequences of vertices as expected. The \textit{size} of the strategy is the size $|M|$ of its Moore machine. Note that a strategy is memoryless when $|M| = 1$.

Given a strategy $\sigma_i$ of $\playerI$, we say that a play $\pi = v_0 v_1 \ldots$ of $G$ is \emph{consistent} with $\sigma_i$ if $v_{k+1} = \sigma_i(v_0 \ldots v_k)$ for all $k \in \N$ such that $v_k \in V_i$. Consistency is naturally extended to histories in a similar fashion. Given an \emph{initial vertex} $v_0$, and a strategy $\sigma_i$ of each player $\playerI$, we have a unique play consistent with both strategies. This play is called the \emph{outcome} of the game and is denoted by $\Out(v_0,\sigma_1,\sigma_2)$.

\paragraph{\bf Objectives and winning sets.} Let $G = (V_1,V_2,E)$ be a game structure. An \emph{objective for $\playerOne$} is a set of plays $\Obj \subseteq \Plays(G)$. A play $\pi$ is \emph{winning} for $\playerOne$ if $\pi \in \Obj$, and losing otherwise (i.e., winning for $\playerTwo$). We thus consider \emph{zero-sum} games such that the objective of player $\playerTwo$ is $\overline{\Obj} = \Plays(G) \ssetminus \Obj$. In the following, we always take the point of view of $\playerOne$ by assuming that $\Obj$ is his objective, and we denote by $(G,\Obj)$ the corresponding \emph{game}. Given an initial vertex $v_0$ of a game $(G,\Obj)$, a strategy $\sigma_1$ for $\playerOne$ is \emph{winning} from $v_0$ if $\Out(v_0,\sigma_1,\sigma_2) \in \Obj$ for all strategies $\sigma_2$ of $\playerTwo$. Vertex $v_0$ is also called \emph{winning} for $\playerOne$ and the \emph{winning set} $\WinG{1}{\Obj}{G}$ is the set of all his winning vertices. Similarly the winning vertices of $\playerTwo$ are those from which $\playerTwo$ can ensure to satisfy his objective $\overline{\Obj}$ against all strategies of $\playerOne$, and $\WinG{2}{\overline{\Obj}}{G}$ is his winning set. If $\WinG{1}{\Obj}{G} \cup \WinG{2}{\overline{\Obj}}{G} = V$, we say that the game is \emph{determined}. It is known that every turn-based game with a Borel objective is determined~\cite{Martin75}. This in particular applies to the objectives studied in this paper.

\paragraph{\bf Decision problem.} Given a game $(G, \Obj)$ and an initial vertex $v_0$, we want to decide whether $\playerOne$ has a winning strategy from $v_0$ for the objective $\Obj$ or not (in which case, $\playerTwo$ has one for $\overline{\Obj}$). We want to study the complexity class of this decision problem as well as the memory requirements of winning strategies of both players. In this paper, we focus on several variants of the parity objective and we consider two settings: the \emph{one-dimension} case with one objective~$\Obj$ and the \emph{multi-dimension} case with the intersection of several objectives $\cap_{m=1}^n \Obj_m$.

\paragraph{\bf Parity objective.} 
Let $G$ be a game structure. Let $\pi$ be a play, we define $\Inf(\pi)$ as the set of vertices seen infinitely often in $\pi$. Formally, $\Inf(\pi) = \{v \in V \mid \forall\, k \geq 0,\: \exists\, l \geq k,\: \pi[l] = v \}$. Given a \emph{priority function} $p\colon V \rightarrow \{0,1,\ldots, d\}$ that maps every vertex to an integer priority where $d$ is even and $d \leq |V| +1$ (w.l.o.g.), the \emph{parity objective} $\Par(p)$ asks that of the vertices that are visited infinitely often, the smallest priority be even. Formally, the parity objective is defined as
\begin{eqnarray*}
\Par(p)= \{ \pi \in \Plays(G) \mid \min_{v \in \Inf(\pi)} p(v) \text{ is even} \}.
\end{eqnarray*}
As smallest even priorities have a specific role in parity objectives, we define a partial order $\stronger$ on priorities as follows. For $c,c' \in \{0,\ldots,d\}$, we have $c \stronger c'$ if and only if $c$ is even and $c \leq c'$. In this case we say that $c$ is \smaller\ than $c'$.

Deciding the winner in parity games is known to be in $\sf UP \cap coUP$~\cite{Jurdzinski98} and whether a polynomial-time algo\-rithm exists is a long-standing open question (e.g.,~\cite{Zielonka98,Jurdzinski00,DBLP:journals/siamcomp/JurdzinskiPZ08,DBLP:conf/fsttcs/Schewe07}). Memoryless strategies suffice for both players to win in parity games~\cite{EmersonJ88,Zielonka98}. The multi-dimension case is studied in~\cite{DBLP:conf/fossacs/ChatterjeeHP07}, with $n$ different priority functions $p_m$, $m \in \{1, \ldots, n\}$, and objective $\cap_{m=1}^n \Obj_m$ such that each $\Obj_m$ is the parity objective defined for~$p_m$. 
Deciding if $\playerOne$ wins in those games is $\sf coNP$-complete, exponential-memory strategies are necessary and sufficient for $\playerOne$, and memoryless strategies suffice for $\playerTwo$~\cite{BuhrkeLV96,DziembowskiJW97,Hor05-GDV,PitermanP06}.

\paragraph*{\bf Other useful objectives.} 

We recall some useful results for several classical objectives. Let $G$ be a game structure and $U \subseteq V$ be a set of vertices. A \emph{reachability} objective $\Reach(U)$ asks to visit a vertex of $U$ at least once, whereas a \emph{safety} objective $\Safe(U)$ asks to visit no vertex of $V \setminus U$. Deciding the winner in reachability games and safety games is known to be $\sf P$-complete with an algorithm in time $O(|V| + |E|)$, and memoryless winning strategies suffice for both objectives and both players~\cite{Beeri80,2001automata,Immerman81}. A \emph{B\"uchi} objective $\Buchi(U)$ asks to visit a vertex of $U$ infinitely often, whereas a \emph{co-B\"uchi} objective $\CoBuchi(U)$ asks to visit no vertex of $V \setminus U$ infinitely often. Deciding the winner in B\"uchi games and co-B\"uchi games is also $\sf P$-complete with an algorithm in time $O(|V|^2)$, and memoryless strategies also suffice for both objectives and both players~\cite{ChatterjeeH14,EmersonJ91,2001automata}. 
Finally, given $U_1,\ldots,U_n$ a family of $n$ subsets of $V$, a \emph{generalized reachability} objective $\GenReach(U_1, . . . , U_n) =  \cap_{m=1}^n \Reach(U_m)$ asks to visit a vertex of $U_m$ at least once, for each $m \in \{1,\ldots,n\}$. Deciding the winner in generalized reachability games is $\sf PSPACE$-complete and exponential-memory strategies are necessary and sufficient for both players~\cite{FijalkowH13}.

\section{Adding time bounds to parity games} \label{sec:objectives}

In this section we introduce the two approaches discussed in this paper: \textit{window parity} (\textsf{WP}) and \textit{parity-response} (\textsf{PR}) games. In Section~\ref{subsec:defs}, we formally define the related objectives. Then, Section~\ref{subsec:decomposition} presents an interesting decomposition of winning plays for the \textsf{WP} objective that will prove to be a useful mechanism to establish solving algorithms. Finally, in Section~\ref{subsec:relationship}, we establish inclusions and equivalences between several variants of the \textsf{WP} and \textsf{PR} objectives.

\subsection{Window parity and parity-response objectives}
\label{subsec:defs}

As stated in Section~\ref{sec:intro}, the intuition for both approaches is as follows. The parity-response objective asks that every priority be followed by a \smaller\ priority in a bounded number of steps. In a window parity game, a \emph{window} with a bounded size is sliding along the play, and one asks to find a \smallest\footnote{Notice the difference: smallest vs.~smaller.} priority inside this window, and this for all positions along the play. We derive four variants for each of these objectives, according to whether the bound is given as a parameter or not (\emph{fixed} or \emph{bounded} variant), and whether the objective must be satisfied directly or eventually (\emph{direct} or \emph{undirect} variant). The \textit{undirect} variants are thus prefix-independent.\footnote{An objective $\Obj$ is \textit{prefix-independent} if for any play $\pi = \rho \pi'$, it holds that $\pi \in \Obj \iff \pi' \in \Obj$.} Formally:

\begin{definition}\label{def:obj}
Given a game structure $G = \Gdev$, a priority function $p\colon V \rightarrow \{0,1,\ldots, d\}$, and a \emph{bound} $\lambda \in \Nzero$, we define the eight following objectives:
\begin{eqnarray*}
\DirPR(\lambda,p) &=& \{ \pi \in \Plays(G) \mid \forall\, j \geq 0,\; \exists\, l \in \{0,\ldots,\lambda-1\},\; p(\pi[j+l]) \stronger p(\pi[j]) \},\\
\DirFWP(\lambda,p) &=& \{ \pi \in \Plays(G) \mid \forall\, j \geq 0,\; \exists\, l \in \{0,\ldots,\lambda-1\},\; \forall\, k \in \{0,\ldots,l\},\; p(\pi[j+l]) \stronger p(\pi[j+k]) \},
\end{eqnarray*}
and given ${\sf X} \in \{\sf{PR,WP}\},$
\begin{eqnarray*}
{\sf FixX}(\lambda,p) &=& \{ \pi \in \Plays(G) \mid \exists\, i \geq 0,\; \pi[i,\infty] \in {\sf DirFixX}(\lambda,p)\}, \\
{\sf DirBndX}(p) &=& \{ \pi \in \Plays(G) \mid \exists\, \lambda \in \Nzero,\; \pi \in {\sf DirFixX}(\lambda,p) \},\\
{\sf BndX}(p) &=& \{ \pi \in \Plays(G) \mid \exists\, i \geq 0,\; \pifactor{i}{\infty} \in {\sf DirBndX}(p) \}.
\end{eqnarray*}
\end{definition}

Thus, in the direct fixed parity-response objective $\DirPR(\lambda,p)$, for all positions $j \geq 0$, the priority $p(\pi[j])$ must be followed by a \smaller\ priority $p(\pi[j+l])$ within at most $\lambda-1$ steps. The undirect fixed variant $\PR(\lambda,p)$ asks this objective to be satisfied eventually (i.e., for all positions $j \geq i$, for some~$i$). The direct bounded variant $\DirBP(p)$ (resp.~the undirect bounded variant $\BP(p)$) asks for the existence of a bound $\lambda$ for which $\DirPR(\lambda,p)$ (resp.~$\PR(\lambda,p)$) is satisfied.

In the case of window parity objectives, we rather call $\lambda$ the \emph{window size}. Given a play $\pi = v_0v_1 \ldots$, a \emph{$\lambda$-\window{j}} is a window of size $\lambda$ placed along $\pi$ from position $j$ to $j + \lambda-1$. The direct fixed window parity objective $\DirFWP(\lambda,p)$ asks that for all $j \geq 0$, inside the $\lambda$-\window{j}, one can find a priority $p(\pi[j+l])$ with $l \leq \lambda - 1$ that is the \smallest\ one in $\pifactor{j}{j+l}$. When the window size $\lambda$ is clear from the context, we drop the prefix $\lambda$ and simply talk about windows instead of $\lambda$-windows.

\begin{example} \label{ex:1} We illustrate the previous definitions on a simple example of one-player game, where all vertices belong to $\playerOne$ (see Figure~\ref{fig:Firstexample}). In this example and in the sequel, the priority $p(v)$ is always put under vertex~$v$, and circle (resp. square) vertices all belong to $\playerOne$ (resp. $\playerTwo$).
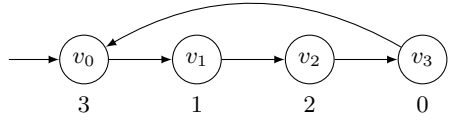
\begin{figure}[h]
\centering
  \begin{tikzpicture}[scale=5]
    \everymath{\footnotesize}
    \draw (0,0) node [circle, draw] (A) {$v_0$};
    \draw (0.3,0) node [circle, draw] (B) {$v_1$};
    \draw (0.6,0) node [circle, draw] (C) {$v_2$};
    \draw (0.9,0) node [circle, draw] (G) {$v_3$};
    
    \draw (0,-.12) node[] (D) {$3$};
	\draw (0.3,-.12) node[] (E) {$1$};    
    \draw (0.6,-.12) node[] (F) {$2$};
    \draw (0.9,-.12) node[] (H) {$0$};
    
    \draw[->,>=latex] (A) to (B);
    \draw[->,>=latex] (B) to (C);
    \draw[->,>=latex] (C) to (G);
    
    \draw[->,>=latex] (G) to[bend right] (A);
    
	\path (-0.2,0) edge [->,>=latex] (A);    
    
    \end{tikzpicture}
    \caption{A simple example of one-player-game: all vertices belong to $\playerOne$.} 
\label{fig:Firstexample}
\end{figure}
\noindent
In the game of Figure~\ref{fig:Firstexample}, there is a unique play from the initial vertex $v_0$ equal to $\pi = (v_0v_1v_2v_3)^{\omega}$. On the one hand, we have that $\pi \in \DirPR(\lambda,p)$ for $\lambda = 3$. Indeed, the odd priority $3$ (resp. $1$) is followed by the even priority $2$ (resp. $0$) in exactly $\lambda - 1 = 2$ steps, whereas the even priority $2$ (resp. $0$) is ``followed" by itself in $0$ steps. Similarly, $\pi$ also belongs to the three variants $\PR(\lambda,p)$, $\DirBP(p)$ and $\BP(p)$. On the other hand, $\pi \not\in \DirFWP(3,p)$. Indeed, in the $3$-\window{0}, there is no $l \in \{0,1,2\}$ such that $p(v_l)$ is the \smallest\ priority in $\pifactor{0}{l}$ because $p(v_0)$ and $p(v_1)$ are odd, and $p(v_2)$ is even but $p(v_2) = 2 \not\stronger 1 = p(v_1)$. However, one can check that $\pi \in \DirFWP(4,p)$, and it also belongs to $\FWP(4,p)$, $\DirBWP(p)$ and $\BWP(p)$.\hfill$\triangleleft$
\end{example}

\paragraph{\bf Links with finitary and window objectives.} Before continuing, we clarify the relation between the objectives introduced in Definition~\ref{def:obj} and the objectives studied in~\cite{ChatterjeeHH09,Chatterjee0RR15,Weinert016}.

\begin{remark}\label{rem:finitary}
In \cite{ChatterjeeHH09}, Chatterjee et al.~study the \emph{bounded parity} (resp. \emph{finitary parity}) objective which asks for the existence of a bound $\lambda$ such that for all positions $j \geq 0$ (resp. $j \geq i$, for some $i$) with an odd priority $p(\pi[j])$, the \emph{minimum} number of steps $l$ to see a \smaller\ priority $p(\pi[j+l])$ does not exceed $\lambda$. These two objectives are respectively equal to the $\DirBP(p)$ and $\BP(p)$ objectives. Indeed asking for a minimum number of steps is not a restriction with respect to our definition, and if $p(\pi[j])$ is even then with $l=0$ we trivially have $p(\pi[j+l]) \stronger p(\pi[j])$. The $\PR(\lambda, p)$ objective can be seen as a particular case of the problem studied in~\cite{Weinert016}, namely the synthesis of so-called optimal strategies in \textit{parity games with costs}, hence $\PR(\lambda, p)$ games are in $\sf PSPACE$. Careful inspection of~\cite{Weinert016} reveals that the $\sf PSPACE$-membership also holds for the direct variant $\DirPR(\lambda, p)$, and that the $\sf PSPACE$-hardness proof can easily be adapted to our sub-cases ($\PR(\lambda, p)$ and $\DirPR(\lambda, p)$).
\end{remark}

\begin{remark}\label{rem:fixwp}
The $\sf WP$ objective and its variants are inspired by the \emph{window mean-payoff} objective and its variants introduced in \cite{Chatterjee0RR15}. In that paper, the classical mean-payoff objective, that requires the average weight to be non-negative in the long run (i.e., at the limit), is replaced by the need for a non-negative average weight inside every bounded window along the play. Here the classical parity objective $\Par(p)$, asking the smallest priority seen infinitely often to be even, is replaced by asking the smallest priority to be even in every bounded window.
\end{remark}

\subsection{Decomposition of plays for the window parity objective}
\label{subsec:decomposition}

We introduce additional terminology concerning the ${\sf WP}$ objective. Let $\pi = v_0v_1 \ldots$ be a play and consider an \textit{unbounded} \window{j}. Essentially, it represents the suffix $\pi[j,\infty]$. We say that this window is \textit{closed} in position $j+k$, for $k \geq 0$, if there exists $l \in \{0, \ldots, k\}$ such that $p(\pi[j+l])$ is the \smallest\ priority in $\pifactor{j}{j+l}$. For the \textit{smallest} such $l$, we say that the \window{j} \textit{closes} in $j+l$. Observe that if $p(\pi[j])$ is even, the \window{j} closes immediately.  On the opposite, the unbounded \window{j} is still \textit{open} in position $j+k$ if no such $l$ exists. Observe that a closed window stays closed forever, and a window can stay open forever if no \smallest\ priority is ever found in prefixes of $\pi[j,\infty]$.

Now, for the objectives presented in Definition~\ref{def:obj}, we are especially interested in windows of bounded size. Let $\lambda \in \mathbb{N}_0$. We say that the \window{j} is \emph{\good} if it closes in a position $j+l$ such that $l < \lambda$, that is, if it closes in at most $(\lambda-1)$ steps.\footnote{This is equivalent to saying that the $\lambda$-\window{j} is closed.} If this is the case, then we also say that the history $\pifactor{j}{j+l}$ is $\lambda$-good.
On the contrary, we say that the \window{j} is \emph{$\lambda$-\bad} if it is still open in position $j + (\lambda - 1)$.\footnote{This is equivalent to saying that the $\lambda$-\window{j} is open.}
\begin{figure}[thb]
\centering
\begin{tikzpicture}[scale=0.8]

    \draw (-1.6,0) node [circle] (A) {$\ldots$};
    \draw (0,0) node [circle, draw,inner sep=4pt,minimum width=13mm] (B) {$\pi[j]$};
    \draw (2.4,0) node [circle, draw,,inner sep=1pt] (C) {};
    \draw (4,0) node [circle] (D) {$\ldots$};
    \draw (5.6,0) node [circle, draw,inner sep=1pt] (E) {};
    \draw (8,0) node [circle, draw,inner sep=2pt,minimum width=13mm] (F) {$\pi[j+l]$};
    \draw (10.4,0) node [circle] (G) {$\ldots$};
    
    \draw (A) -- (B);
	\draw (B) -- (C);
	\draw (C) -- (D);
	\draw (D) -- (E);
	\draw (E) -- (F);
	\draw (F) -- (G);
    
    \draw (B) to[bend left] node[above,pos =0.1,xshift=-3mm] {\good} (F);
    \draw (C) to[bend left] node[above,pos =0.2,xshift=-3mm] {\good} (F);
    \draw (E) to[bend left] node[above,pos =0.1,xshift=-3mm] {\good~~~} (F);

\end{tikzpicture}
\caption{History $\pifactor{j}{j+l}$ is \good\ implies that for all $j' \in \{j+1, \ldots, j+l\}$, $\pifactor{j'}{j+l}$ is also \good.}
\label{fig:iclosed}
\end{figure}
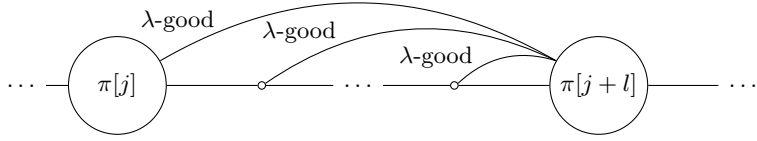

An interesting property is the following one: if $\pifactor{j}{j+l}$ is \good, then $\pifactor{j'}{j+l}$ is also \good\ for all $j' \in \{j+1, \ldots, j+l\}$ (see Figure~\ref{fig:iclosed}). In particular, each window at position $j' \in \{j+1, \ldots, j+l\}$ is closed in position $j+l$.

The next lemma will prove useful later on. Intuitively, it states that any play that is winning for a \textsf{WP} objective can be seen as a succession of $\lambda$-good histories which can serve as the basis for a corresponding decomposition.

\begin{lemma} \label{lem:goodDec}
A play $\pi$ belongs to $\FWP(\lambda, p)$ (resp. $\DirFWP(\lambda, p)$) if and only if there exists an increasing sequence of indices $(k_i)_{i \geq 0}$ (resp. with $k_0 = 0$) such that for each $i$, $\pifactor{k_i}{k_{i+1}-1}$ is \good.
\end{lemma}

\noindent
When such a sequence $(k_i)_{i \geq 0}$ with $k_0 = 0$ exists for a play $\pi$, we say that it is a \emph{\GoodDec{\lambda}} of $\pi$. When $k_0 \geq 0$, it is called an \emph{\EGoodDec{\lambda}}.

\begin{proof}[of Lemma~\ref{lem:goodDec}]
We only give the proof for the direct variant. Suppose that $\pi$ belongs to $\DirFWP(\lambda, p)$. Then by definition of $\DirFWP(\lambda, p)$, the \window{k_0 = 0} is \closed\ in at most $(\lambda-1)$ steps, i.e.,  there exists $k_1> k_0$ such that $\pifactor{k_0}{k_1-1}$ is \good. Next, as the \window{k_1} is \closed\ in at most $(\lambda-1)$ steps, there exists $k_2 > k_1$ such that $\pifactor{k_1}{k_2-1}$ is \good\ also. This leads to a \GoodDec{\lambda} $(k_i)_{i \geq 0}$ of $\pi$. 

Conversely, if there exists a \GoodDec{\lambda} $(k_i)_{i \geq 0}$ of $\pi$, then for all $i \geq 0$, $\pifactor{k_i}{k_{i+1}-1}$ is \good, as well as $\pifactor{k'}{k_{i+1}-1}$ for all $k' \in \{k_i+1,\ldots,k_{i+1}-1\}$. It follows that $\pi$ belongs to $\DirFWP(\lambda, p)$.
\qed\end{proof}

\subsection{Relationship between objectives}
\label{subsec:relationship}

We now detail the inclusions and equalities between the various objectives introduced in Definition~\ref{def:obj} as well as with the parity objective.

\begin{proposition} \label{prop:inclusions}
Let $G = (V_1, V_2, E)$ be a game structure and $p$ be a priority function. Let $\lambda \in \Nzero$.
\begin{enumerate}
\item For all ${\sf X} \in \{\PR(\lambda, p), \FWP(\lambda, p), \BP(p), \BWP(p)\}$, ${\sf DirX}$ $\subseteq {\sf X}$. \label{item:direct}
\item For all ${\sf X} \in \{\sf{PR},\sf{WP}\}$, $({\sf Dir}){\sf FixX}(\lambda, p)$ $\subseteq ({\sf Dir}){\sf BndX}(p)$. \label{item:bounded}
\item For all $\lambda' > \lambda$, for all ${\sf X} \in \{\PR,\FWP\}$, $({\sf Dir}){\sf X}(\lambda,p) \subseteq ({\sf Dir}){\sf X}(\lambda',p)$.\label{item:monotone}
\item ${\sf(Dir)}\FWP(\lambda,p) \subseteq {\sf(Dir)}\PR(\lambda,p)$. \label{item:WPcPR}
\item $({\sf Dir})\PR(\lambda,p) \subseteq ({\sf Dir})\FWP(\frac{d}{2} \cdot \lambda,p)$. \label{item:encadrement}
\item $({\sf Dir})\BP(p) = ({\sf Dir})\BWP(p)$.\label{item:samebnd}
\item ${\sf(Dir)}\BWP(p)\subseteq \Par(p)$. \label{item:BWPcPar}
\end{enumerate}
\end{proposition}

Before proving those statements formally, we give an intuitive explanation of Item~\ref{item:encadrement}, as it is the most interesting one technically. Assume we have a play $\pi \in \DirPR(\lambda,p)$, like the one depicted in Figure~\ref{fig:fromPRtoWP}. Since it satisfies objective $\DirPR(\lambda,p)$, we know that each odd priority is followed by a smaller even priority in at most $(\lambda-1)$ steps. We argue that we can give a $\lambda'$-good decomposition of this play for some $\lambda' \leq \frac{d}{2}\cdot \lambda$, hence that it belongs to $\DirFWP(\frac{d}{2} \cdot \lambda,p)$. The key idea is depicted in Figure~\ref{fig:fromPRtoWP}. Let $c_1$ be an odd priority. It must be followed by a \smaller\ priority $c'_1$ in at most $(\lambda - 1)$ steps. If $c'_1$ is the minimal priority encountered from $c_1$ to $c'_1$, then we are done as the corresponding history is $\lambda$-good. Assume it is not, then there exists $c_2$ between $c_1$ and $c'_1$ such that $c_2$ is odd and $c'_1 \not\stronger c_2$. But again, $c_2$ must be followed by $c'_2 \stronger c_2$ in at most $(\lambda - 1)$ steps. Repeating this argument, we obtain that $c_1$ is followed by a priority $c'_k$ in strictly less than $\frac{d}{2}\cdot \lambda$ steps\footnote{Actually, $\frac{d}{2}\cdot(\lambda-1) + 1$ but we use the simpler bound $\frac{d}{2}\cdot\lambda$ from now on for the sake of readability.} (as there are $\frac{d}{2}$ odd priorities and each of them is answered in $(\lambda-1)$ steps) such that $c'_k$ is even and smaller than all priorities encountered from $c_1$ to $c'_k$. Therefore, the corresponding history is $\lambda'$-good for $\lambda' \leq \frac{d}{2}\cdot\lambda$. As this argument can be repeated from the vertex following priority $c'_k$, we obtain a $\lambda'$-good decomposition of the play, which implies that it satisfies $\DirFWP(\lambda',p)$ as claimed.

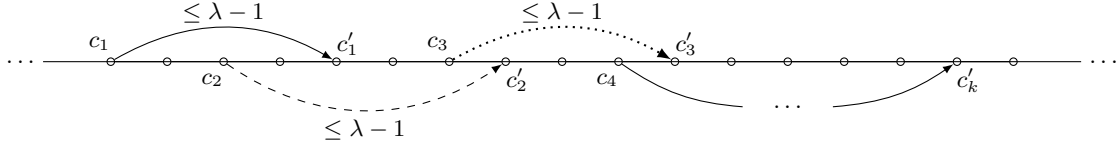
\begin{figure}[t]
\centering
\begin{tikzpicture}[scale=3]

	\draw (-0.4,0) node (R) {$\ldots$};
	\draw (4.4,0) node (R) {$\ldots$};
	\draw (0,0) node [circle, draw,inner sep=1pt] (A) {};
	\draw (-0.05,0.08) node [] (AA) {$c_1$};
	\draw (1.05,0.08) node [] (AAb) {$c'_1$};
	\draw (0.25,0) node [circle, draw,inner sep=1pt] (B) {};
	\draw (0.5,0) node [circle, draw,inner sep=1pt] (C) {};
	\draw (0.75,0) node [circle, draw,inner sep=1pt] (D) {};
	\draw (0.75,-0.15) node [] (DD) {};
	\draw (1,0) node [circle, draw,inner sep=1pt] (E) {};
	\draw (1.25,0) node [circle, draw,inner sep=1pt] (F) {};
	\draw (1.5,0) node [circle, draw,inner sep=1pt] (G) {};
	\draw (1.75,0) node [circle, draw,inner sep=1pt] (H) {};
	\draw (2,0) node [circle, draw,inner sep=1pt] (I) {};
	\draw (2,-0.15) node [] (II) {};
	\draw (2.25,0) node [circle, draw,inner sep=1pt] (J) {};
	\draw (2.5,0) node [circle, draw,inner sep=1pt] (K) {};
	\draw (2.5,-0.15) node [] (KK) {};
	\draw (2.75,0) node [circle, draw,inner sep=1pt] (L) {};
	\draw (3,0) node [circle, draw,inner sep=1pt] (M) {};
	\draw (3.25,0) node [circle, draw,inner sep=1pt] (N) {};
	\draw (3.5,0) node [circle, draw,inner sep=1pt] (O) {};
	\draw (3.75,0) node [circle, draw,inner sep=1pt] (P) {};
	\draw (4,0) node [circle, draw,inner sep=1pt] (Q) {};
	\draw (4,-0.15) node [] (QQ) {};
	\draw (3,-0.2) node (R) {$\ldots$};
	\draw (1.45,0.08) node [] (l) {$c_3$};
	\draw (2.55,0.08) node [] (l) {$c'_3$};
	\draw (0.45,-0.08) node [] (l) {$c_2$};
	\draw (1.80,-0.08) node [] (l) {$c'_2$};
	\draw (2.2,-0.08) node [] (l) {$c_4$};
	\draw (3.80,-0.08) node [] (l) {$c'_k$};
	
	\draw[->,>=latex] (A) to[bend left=30] (E);
	\draw[->,>=latex,dashed] (C) to[bend right=35] (H);
	\draw[->,>=latex,dotted,thick] (G) to[bend left=30] (K);
	\draw (J) to[bend right=15] (2.8,-0.2);
	\draw[->,>=latex] (3.2,-0.2) to[bend right=15] (P);
	\draw (0.5,0.22) node (l1) {$\leq \lambda - 1$};
	\draw (2,0.22) node (l1) {$\leq \lambda - 1$};
	\draw (1.125,-0.3) node (l1) {$\leq \lambda - 1$};

	\draw (-0.3,0) -- (4.3,0);
	\draw (A) -- (B);
	\draw (B) -- (C);
	\draw (C) -- (D);
	\draw (D) -- (E);
	\draw (E) -- (F);
	\draw (F) -- (G);
	\draw (G) -- (H);
	\draw (H) -- (I);
	\draw (I) -- (J);
	\draw (J) -- (K);
	\draw (K) -- (L);
	\draw (L) -- (M);
	\draw (M) -- (N);
	\draw (N) -- (O);
	\draw (O) -- (P);
	\draw (P) -- (Q);

\end{tikzpicture}
\caption{Illustration of inclusion $({\sf Dir})\PR(\lambda,p) \subseteq ({\sf Dir})\FWP(\frac{d}{2} \cdot \lambda,p)$. Priorities $c_i$ are all odd and decreasing ($c_{i+1} < c_i$), priorities $c'_i$ are even and such that $c'_i \stronger c_i$ but $c'_i \not\stronger c_{i+1}$. Eventually, a priority $c'_k$ is reached such that the history from $c_1$ to $c'_k$ is $\lambda'$-good for some $\lambda'$. Furthermore $k \leq \frac{d}{2}$, hence $\lambda' \leq \frac{d}{2}\cdot\lambda$.}
\label{fig:fromPRtoWP}
\end{figure}

\begin{proof} The proof is given for the direct variants only, as it is similar for the undirect variants. The first four items immediately follow from the definitions.

We now prove Item~\ref{item:encadrement}. Let $\pi \in \DirPR(\lambda,p)$ and let us show that $\pi \in \DirFWP(\lambda',p)$ with $\lambda' = \frac{d}{2} \cdot \lambda$. Given $j_1 \in \N$, we are looking for an integer $l' \in \{0,\ldots,\lambda'-1\}$ such that $p(\pi[j_1+l'])$ is the \smallest\ priority in $\pifactor{j_1}{j_1+l'}$. If $p(\pi[j_1])$ is even then take $l' = 0$. Otherwise, as $\pi \in \DirPR(\lambda,p)$, (*)  there exists $l_1 \in \{0,\ldots,\lambda-1\}$ (that we take minimal) such that $p(\pi[j_1+l_1]) \stronger p(\pi[j_1])$. If $p(\pi[j_1+l_1])$ is the \smallest\ priority in $\pifactor{j_1}{j_1+l_1}$, then take $l' = l_1$. Otherwise, the smallest priority in $\pifactor{j_1}{j_1+l_1}$ is odd (by minimality of $l_1$) and strictly smaller than $p(\pi[j_1])$. Let $j_2$ be a position in $\pifactor{j_1}{j_1+l_1}$ of this smallest priority and repeat (*) for $j_2$. There now exists a minimal $l_2 \in \{0,\ldots,\lambda-1\}$ such that $p(\pi[j_2+l_2]) \stronger p(\pi[j_2])$. Again, if $p(\pi[j_2+l_2])$ is the \smallest\ priority in $\pifactor{j_2}{j_2+l_2}$, then by definition of $j_2$, $p(\pi[j_2+l_2])$ is also the \smallest\ priority in $\pifactor{j_1}{j_2+l_2}$. Therefore we can take $l' = j_2+l_2-j_1$. Otherwise, the smallest priority $p(\pi[j_3])$ in $\pifactor{j_2}{j_2+l_2}$ is odd and strictly smaller than $p(\pi[j_2])$, and we now repeat (*) for $j_3$, also. As $p(\pi[j_1]), p(\pi[j_2]), p(\pi[j_3]) \ldots$ is a decreasing sequence of odd priorities with $j_2 - j_1 \leq \lambda-1$ and $j_3 - j_1 \leq 2(\lambda-1)$, we end the process with $p(\pi[j_\frac{d}{2} + l_\frac{d}{2}]) = 0$ in the worst case. It follows that $p(\pi[j_\frac{d}{2}+l_\frac{d}{2}])$ is the \smallest\ priority in $\pifactor{j_1}{j_\frac{d}{2} + l_\frac{d}{2}}$ and we take $l' = j_\frac{d}{2} + l_\frac{d}{2} - j_1$. Notice that $l' < \frac{d}{2} \cdot \lambda$.

Item~\ref{item:samebnd} is a direct consequence of Items~\ref{item:WPcPR} and~\ref{item:encadrement}.

In order to prove Item~\ref{item:BWPcPar}, let $\pi \in \DirBWP(p)$ and suppose that $\pi \not\in \Par(p)$. Thus among the vertices that are visited infinitely often, the least priority, call it $c$, is odd. Let $j \in \N$ be such that $p(\pi[j]) = c$ and every vertex in $\pifactor{j}{\infty}$ belongs to $\Inf(\pi)$. It follows that for all $l \geq 0$, $p(\pi[j+l]) \not\stronger p(\pi[j])$. This is in contradiction with $\pi$ belonging to $\DirBWP(p)$ as the window at position $j$ would never close, and therefore we conclude that $\pi \in \Par(p)$.
\qed
\end{proof}

From the inclusions $\Obj \subseteq \Obj'$ of Proposition~\ref{prop:inclusions}, we immediately derive the inclusions $\WinG{1}{\Obj}{G} \subseteq \WinG{1}{\Obj'}{G}$. It yields two interesting observations.
\begin{enumerate}
\item By Items~\ref{item:direct},~\ref{item:bounded},~\ref{item:samebnd}, and~\ref{item:BWPcPar}, we see that all \textsf{WP} and \textsf{PR} variants provide conservative approximations of the classical parity objective. While the bounded variants will all be shown to be in $\sf P$ (Theorem~\ref{thm:finitaryparity}), \textit{for the fixed variants, the \textsf{WP} objective will be the only polynomial-time alternative} (Theorem~\ref{prop:WPonedim}).
\item The \textit{fixed \textsf{WP} objective can actually be used to provide both an under-approximation and an over-approximation of the fixed \textsf{PR} one in polynomial-time} (Theorem~\ref{prop:WPonedim}), as witnessed by Items~\ref{item:WPcPR} and~\ref{item:encadrement}. This is particularly interesting since the fixed \textsf{PR} objective itself is $\sf PSPACE$-complete (Theorem~\ref{prop:PRonedim}).
\end{enumerate}
Notice that the inclusions of Proposition~\ref{prop:inclusions} are strict in general. This is also the case when one replaces the objectives by the winning sets of $\playerOne$ for these objectives. We illustrate this on the following example.

\begin{figure}[th]
\centering
  \begin{tikzpicture}[scale=5]
    \everymath{\footnotesize}
    \draw (0,0) node [circle, draw] (A) {$v_0$};
    \draw (0.45,0) node [rectangle, draw,inner sep=4pt] (B) {$v_1$};
    \draw (0.9,0) node [circle, draw] (C) {$v_2$};
    
    \draw (0,-.12) node[] (D) {$1$};
	\draw (0.45,-.12) node[] (E) {$2$};    
    \draw (0.9,-.12) node[] (F) {$0$};
    
    \draw[->,>=latex] (A) to (B);
    \draw[->,>=latex] (B) to (C);
    
    \draw[->,>=latex] (C) to[bend right] (A);
    \draw[->,>=latex] (B) .. controls +(315:0.3cm) and +(225:0.3cm) .. (B);
    
	\path (-0.2,0) edge [->,>=latex] (A);    
    
    \end{tikzpicture}
\caption{The initial vertex $v_0$ is winning for the parity objective but is losing for all variants of objectives \textsf{WP} and \textsf{PR}: $\playerTwo$ has the possibility to use the self-loop on $v_1$ to delay for an arbitrarily long time the response $0$ to priority $1$, and can do so repeatedly using the other loop, thus defeating both direct and undirect variants of the objectives.}
\label{fig:Secondexample}
\end{figure}
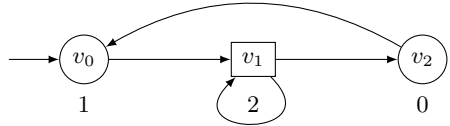

\begin{example} \label{ex:2}
Recall the game in Figure~\ref{fig:Firstexample}. We showed in Example~\ref{ex:1} that the only possible outcome, $\pi = (v_0v_1v_2v_3)^{\omega}$, belongs to $\DirPR(\lambda,p) \setminus \DirFWP(\lambda,p)$ for $\lambda = 3$. This shows that inclusion of Item~\ref{item:WPcPR} in Proposition~\ref{prop:inclusions} is strict. As this game is a one-player game with a unique play, it follows that $\WinG{1}{\DirFWP(3,p)}{G} \subsetneq \WinG{1}{\DirPR(3,p)}{G}$. 

Consider now the game depicted in Figure~\ref{fig:Secondexample}. First, we show that vertex $v_0$ is winning for $\playerOne$ for the parity objective. Indeed, either $\playerTwo$ eventually loops forever on vertex $v_1$, or he visits vertex $v_2$ infinitely often. Then, either priority~$2$ (of $v_1$) is the only one seen infinitely, or priority~$0$ (of $v_2$) is seen infinitely often, which shows that $v_0 \in \WinG{1}{\Par(p)}{G}$. However, $\playerOne$ loses from $v_0$ for the objective $\DirPR(\lambda,p)$ (for all $\lambda \in \Nzero$) since $\playerTwo$ can ensure that priority~$1$ is never followed by a \smaller\ priority by always looping on vertex $v_1$. Thus, $\WinG{1}{\DirPR(\lambda,p)}{G}\subsetneq \WinG{1}{\Par(p)}{G}$ for all $\lambda$, and thus also $\WinG{1}{\DirBP(p)}{G}\subsetneq \WinG{1}{\Par(p)}{G}$. The same strategy permits $\playerTwo$ to win the game when the objective is $\DirFWP(\lambda,p)$ for all $\lambda$, and thus also when the objective is $\DirBWP(p)$.

Let us go further with the game in Figure~\ref{fig:Secondexample} by considering the $\BWP(p)$ and $\BP(p)$ objectives, which are prefix-independent as the parity objective. We can see that $v_0$ is still losing for $\playerOne$ for both objectives. Indeed, consider the infinite-memory strategy $\sigma_2$ of $\playerTwo$ which consists in looping for longer and longer time periods in $v_1$ before going to $v_2$. Formally, let $n$ be a counter initialized to $1$, the strategy $\sigma_2$ loops $n$ times in $v_1$, increments $n$ and goes to $v_2$. Such a strategy increases the number of steps before priority~$1$ is followed by a \smaller\ priority each ``round''. Therefore, there exists no bound $\lambda \in \mathbb{N}_0$ for which the outcome belongs to $\PR(\lambda,p)$ and, similarly, there is no window size $\lambda \in \mathbb{N}_0$ for which the outcome belongs to $\FWP(\lambda,p)$. Essentially, $\playerTwo$ uses his infinite-memory to create ever-increasing delays. This shows that for all $\lambda \in \mathbb{N}_0$, $v_0 \not\in \WinG{1}{\PR(\lambda,p)}{G}$ and $v_0 \not\in \WinG{1}{\FWP(\lambda,p)}{G}$. Consequently, we also have that $v_0 \not\in \WinG{1}{\BP(p)}{G}$ and $v_0 \not\in \WinG{1}{\BWP(p)}{G}$.

Observe that $\sigma_2$ uses \textit{infinite memory}. Actually, $\playerTwo$ cannot win the \textit{undirect bounded} objectives from $v_0$ for any finite-memory strategy~$\sigma'_2$. Indeed, in this case, either the outcome $\pi$ eventually loops on $v_1$ forever, or each time it visits $v_1$, $\pi$ loops on $v_1$ at most $m$ times (for some $m \in \mathbb{N}$ depending on the finite memory of~$\sigma'_2$) and then goes to $v_2$. It follows that the delay is \textit{eventually} bounded in both cases and $\pi \in \PR(m+2,p) \subseteq \BP(p)$ and $\pi \in \FWP(m+2,p) \subseteq \BWP(p)$.\hfill$\triangleleft$
\end{example}

We close this section by establishing that for the sub-case of games with priorities in $\{0, 1, 2\}$, \textsf{WP} and \textsf{PR} objectives coincide.

\begin{lemma}\label{lem:sameobj}
Let $G$ be a game structure and $p \colon V \rightarrow \{0,1,2\}$ be a priority function. For all $\lambda \in \Nzero$,
\[
({\sf Dir})\PR(\lambda,p)  = ({\sf Dir})\FWP(\lambda,p).
\]
\end{lemma}

\begin{proof}
Again, we only give the proof for the direct variant. By Proposition~\ref{prop:inclusions}, we already know that $\DirFWP(\lambda,p) \subseteq \DirPR(\lambda,p)$. To show the other inclusion, let $\pi$ be a play in $\DirPR(\lambda,p)$. Then, for all $j \geq 0$, there exists $l \in \{0,\ldots,\lambda-1\}$ such that $p(\pi[j+l]) \stronger p(\pi[j])$. Now, as $p \colon V \rightarrow \{0,1,2\}$, we have that either $p(\pi[j])$ is even (and $l=0$) or $p(\pi[j]) = 1$ and $p(\pi[j+l]) = 0$. In particular, $p(\pi[j+l])$ is the \smallest\ priority in $\pifactor{j}{j+l}$ showing that $\pi \in  \DirFWP(\lambda,p)$. 
\qed
\end{proof}

\section{One-dimension games} \label{sec:onedim}

We begin our study of \textsf{WP} and \textsf{PR} objectives with one-dimension games: in this setting, there is a unique priority function $p$ and the objective $\Obj$ is a single objective $({\sf Dir}){\sf FixX}$ or $({\sf Dir}){\sf BndX}$ for ${\sf X} \in \{{\sf PR}, {\sf WP}\}$. We first address the \textit{bounded} variants in Section~\ref{subsec:bounded}, then turn to the \textit{fixed} ones in Section~\ref{subsec:fixed}.

\subsection{Bounded variants} \label{subsec:bounded} 

Recall that we know by Proposition~\ref{prop:inclusions} that the bounded variants are equivalent. Furthermore, it is already known that games with objective $({\sf Dir})\BP(p)$ are solvable in polynomial time~\cite{ChatterjeeHH09}. The next theorem sums up the complexity landscape for bounded variants and enrich it by proving $\sf P$-hardness for the associated decision problems. In terms of memory requirements, $\playerOne$ can play without memory whereas Example~\ref{ex:2} already illustrated that $\playerTwo$ requires infinite memory in general. The linear memory bound for $\playerTwo$ and the direct variant was established in~\cite{DBLP:journals/corr/abs-1207-0663}.

\begin{theorem} \label{thm:finitaryparity} 
Let $G = (V_1,V_2,E)$ be a game structure, $v_0$ be an initial vertex, $p$ be a priority function, and $\Obj$ be the objective $\DirBP(p)$ or $\DirBWP(p)$ (resp. $\BP(p)$ or $\BWP(p)$). 
\begin{enumerate}
\item Deciding the winner in $(G,\Obj)$ from $v_0$ is $\sf P$-complete with an algorithm in $\mathcal{O}(|V| \cdot |E|)$ (resp.~$\mathcal{O}(|V|^2 \cdot |E|)$) time, memoryless strategies are sufficient for $\playerOne$, and linear-memory strategies are necessary and sufficient  for $\playerTwo$ (resp.~infinite memory is necessary for $\playerTwo$). \label{item:finitaryparityItem1}
\item $\forall\,\lambda \geq |V|,\; \forall\, \lambda' \geq \frac{d}{2}\cdot|V|$, the winning sets for the objectives $\BP(p)$, $\PR(\lambda,p)$, $\BWP(p)$, and $\FWP(\lambda',p)$
are all equal. \label{item:finitaryparityItem2}
\item The equalities given in Item~\ref{item:finitaryparityItem2} also hold for the direct variants ${\sf (Dir)}$. \label{item:finitaryparityItem3}
\end{enumerate}
\end{theorem}

\begin{proof}
$(a)$ Let us prove Item~\ref{item:finitaryparityItem1}. The $\DirBP(p)$ (resp. $\BP(p)$) objective is studied in \cite{ChatterjeeHH09} and shown to be in $\sf P$ with an algorithm in $\mathcal{O}(|V| \cdot |E|)$ (resp.~$\mathcal{O}(|V|^2 \cdot |E|)$) time. Moreover it is shown that memoryless strategies are sufficient for $\playerOne$, and linear-memory strategies are necessary and sufficient for $\playerTwo$~\cite{DBLP:journals/corr/abs-1207-0663} (resp. infinite memory is necessary for $\playerTwo$). As $({\sf Dir})\BP(p) = ({\sf Dir})\BWP(p)$ by Item~\ref{item:samebnd} of Proposition~\ref{prop:inclusions}, we have the same complexity results and memory requirements for the $({\sf Dir})\BWP(p)$ objectives. 

To complete the proof of Item~\ref{item:finitaryparityItem1}, we have to prove that deciding the winner in $(G,({\sf Dir})\BWP(p))$ is $\sf P$-hard. We begin with the undirect variant. Let $G_r = (V_1,V_2,E_r)$ be a game structure, and consider the reachability objective $\Reach(U)$ with the target set $U \subseteq V$. From $G_r$, we build the game structure $G = (V_1,V_2,E)$ by $(i)$ making the target vertices absorbing with a self-loop and $(ii)$ defining the priority function $p \colon V \rightarrow \{0,1\}$ as follows: $p(v) = 0$ if $v\in U$, and $p(v) = 1$ otherwise. We claim that $\playerOne$ has a winning strategy in $G_r$ from an initial vertex $v_0$ for the $\Reach(U)$ objective if and only if he has a winning strategy for the $\BWP(p)$ objective in $G$ from $v_0$. Indeed, any outcome that never reaches the target set $U$ in $G_r$ is such that for all $j \in \N$ the \window{j} in $G$ stays open forever, i.e., it is $\lambda$-\bad\ for any size $\lambda \in \mathbb{N}_0$. Conversely, any outcome that reaches some $v \in U$ in $j$ steps in $G_r$ has a corresponding outcome in $G$ that reaches $v$ in $j$ steps and then loops on $v$ (recall that $v$ is absorbing), that is, from position $j$ all windows of size $\lambda = 1$ are closed. Thus $v_0$ is winning for the $\Reach(U)$ objective in $G_r$ if and only if $v_0$ is winning for the $\BWP(p)$ objective in $G$. This concludes the proof of $\sf P$-hardness in case of $\BWP(p)$ objective since deciding the winner in reachability games is $\sf P$-complete~\cite{Beeri80,Immerman81}.

The same reduction holds for the direct variant $\DirBWP(p)$. We already know that if $\playerOne$ has a winning strategy for $\BWP(p)$ in $G$, he has one for $\Reach(U)$ in $G_r$. Since $\DirBWP(p) \subseteq \BWP(p)$ by Proposition~\ref{prop:inclusions}, Item~\ref{item:direct}, extending this direction to $\DirBWP(p)$ is trivial. It remains to consider the converse, that is, if $\playerOne$ has a winning strategy for $\Reach(U)$ in $G_r$, he has one for $\DirBWP(p)$ in $G$. Using the arguments from above, we actually see that any winning play $\pi = v_0 v_1 \ldots{}$ in $G_r$ belongs to $\DirFWP(k+1,p)$ in $G$ where $k$ is the first index such that $v_k \in U$. Hence, all winning plays in $G_r$ belong to $\DirBWP(p)$ in $G$.

\medskip
$(b)$ Let us now proceed with the proof of Items~\ref{item:finitaryparityItem2} and~\ref{item:finitaryparityItem3}. Recall that $({\sf Dir})\BP(p) = {(\sf Dir)}\BWP(p)$ by Item~\ref{item:samebnd} of Proposition~\ref{prop:inclusions}. Let us consider the ${\sf PR}$ objectives. It is shown in~\cite{ChatterjeeHH09} that $\WinG{1}{{\sf (Dir)}\BP(p)}{G} = \WinG{1}{{\sf (Dir)}\PR(\lambda,p)}{G}$ with $\lambda = |V|$. By Items~\ref{item:bounded} and~\ref{item:monotone} of Proposition~\ref{prop:inclusions}, we also have equalities of those sets with $\WinG{1}{{\sf (Dir)}\PR(\lambda',p)}{G}$ for all $\lambda' \geq |V|$. We get the required equalities for $\sf WP$ objectives by Items~\ref{item:monotone}, \ref{item:WPcPR}, \ref{item:encadrement} and~\ref{item:samebnd} of Proposition~\ref{prop:inclusions}.  
\qed \end{proof}

\subsection{Fixed variants} \label{subsec:fixed}

The fixed variants are more interesting: the \textsf{PR} and \textsf{WP} approaches yield different results in this setting. We start with the \textsf{PR} one, for which we provide two polynomial-time algorithms for fixed-parameter sub-cases, hence significantly reducing the complexity of the problem (which is $\sf PSPACE$-complete in the general case).

\paragraph*{\bf Parity-response objectives.}

Deciding the winner in $({\sf Dir})\PR$ games was very recently proved to be $\sf PSPACE$-complete~\cite{Weinert016}. As mentioned in Remark~\ref{rem:finitary}, the proof was actually provided for a more general model, but already holds for both $\PR$ and $\DirPR$ games.

\begin{theorem}[\cite{Weinert016}] \label{prop:PRonedim}
Let $G = (V_1,V_2,E)$ be a game structure, $v_0$ be an initial vertex, $p$ be a priority function, and $\Obj$ be the objective $({\sf Dir})\PR(\lambda,p)$ for some $\lambda < \vert V\vert$. Deciding the winner in $(G,\Obj)$ from $v_0$ is $\sf PSPACE$-complete. Moreover, exponential memory is both necessary and sufficient for $\playerOne$, and exponential memory is sufficient for $\playerTwo$ while linear memory is necessary.
\end{theorem}

Observe that the $\sf PSPACE$-hardness only holds for time bounds $\lambda < \vert V\vert$ since we know by Theorem~\ref{thm:finitaryparity}, Items~\ref{item:finitaryparityItem2} and~\ref{item:finitaryparityItem3}, that for larger values, the objectives are equivalent to the bounded variants, hence the corresponding decision problems lie in $\sf P$.
In addition, we focus on the case $\lambda < |V|$: we show in the next theorem that when we fix either the largest priority~$d$ or the bound $\lambda$, the complexity collapses to~$\sf P$. We briefly sketch the two algorithms here (we illustrate the memory bounds in the upcoming Example~\ref{ex:memory}).

First, consider the case where $d$ is fixed. We reduce the $\PR(\lambda, p)$ (resp.~$\DirPR(\lambda,p)$) game to a co-B\"uchi (resp.~safety) game on an extended graph where we keep track of additional information in the vertices. Namely, we keep a vector that represents, for each odd priority $c$, the number of steps since seeing~$c$ without seeing any \smaller\ priority in the meantime. When this number reaches $\lambda$ for any odd priority, we visit a special ``bad vertex'' and then reset the counters in the vector and resume the game. Essentially, winning for $\PR(\lambda, p)$ (resp.~$\DirPR(\lambda,p)$) boils down to eventually (resp.~completely) avoiding those bad vertices, hence to a co-B\"uchi (resp.~safety) game. This extended game has size $\mathcal{O}(|V|\cdot \lambda^{\frac{d}{2}})$ and can be solved in polynomial time since $\lambda < |V|$ (otherwise we use the algorithm for the \textit{bounded} variants presented in Theorem~\ref{thm:finitaryparity}) and $d$ is fixed.

Second, consider the case where $\lambda$ is fixed (and $< |V|$ for the same reason as before). We also reduce the $\PR(\lambda, p)$ (resp.~$\DirPR(\lambda,p)$) game to a co-B\"uchi (resp.~safety) game, but with a different extended graph. Specifically, we here keep track of the last $\lambda$ vertices seen in the original game, and we want to avoid vertices of the extended graph that correspond to histories where an odd priority $c$ is not followed by a priority $c' \stronger c$ within $(\lambda-1)$ steps. Again, this can be expressed as either a co-B\"uchi or a safety objective depending on whether we are interested in the undirect or the direct variant respectively. The extended game has size $\mathcal{O}(|V|^\lambda)$ hence can be solved in polynomial time since $\lambda$ is fixed.

\begin{theorem}\label{thm:fixed}
Let $G = (V_1,V_2,E)$ be a game structure, $v_0$ be an initial vertex, $p \colon V \rightarrow \{0,\ldots,d\}$ be a priority function, and $\Obj$ be the objective $\DirPR(\lambda,p)$ (resp. $\PR(\lambda,p)$) for some $\lambda \in \Nzero$. If either $d$ is fixed or $\lambda$ is fixed, deciding the winner in $(G,\Obj)$ from $v_0$ is in $\sf P$. More precisely,
\begin{enumerate}
\item If $\lambda \geq |V|$, deciding the winner can be done in $\mathcal{O}(|V| \cdot |E|)$ (resp. $\mathcal{O}(|V|^2 \cdot |E|)$) time, memoryless strategies are sufficient for $\playerOne$, and linear-memory strategies are both necessary and sufficient for $\playerTwo$ (resp. infinite memory is necessary for $\playerTwo$). \label{item:fixed1}
\item If $\lambda < |V|$ and $d$ is fixed, deciding the winner can be done in $\mathcal{O}((|V| + |E|) \cdot \lambda^\frac{d}{2})$ (resp. $\mathcal{O}(|V|^2 \cdot \lambda^d)$) time, polynomial-memory strategies with $\mathcal{O}(\lambda^\frac{d}{2})$ memory are sufficient for both players, and memory is necessary even in one-player games. \label{item:fixed2}
\item If $\lambda < |V|$ and $\lambda$ is fixed, deciding the winner can be done in $\mathcal{O}((|V|+|E|)\cdot |V|^{\lambda-1})$ (resp. $\mathcal{O}(|V|^{2\lambda})$) time, polynomial-memory strategies with $\mathcal{O}(|V|^{\lambda-1})$ memory are sufficient for both players, and memory is necessary even in one-player games. \label{item:fixed3}
\end{enumerate}
\end{theorem}

\begin{proof}
$(a)$ The first item directly follows from Theorem~\ref{thm:finitaryparity}.

\medskip
$(b)$ Let us prove Item~\ref{item:fixed2}. Let $\lambda < |V|$. We first consider the undirect variant. From $G$, we construct a game~$G'$ that keeps track in its vertices of the current vertex $v$ of $G$ and whether each seen odd priority $c$ has been followed by a \smaller\ priority within at most $\lambda-1$ steps. To this end, to each odd priority $c$ is associated a counter $l_c \in \{\bot,0,1,\ldots,\lambda - 2\}$ such that during the $\lambda - 1$ last steps $(i)$ either $l_c = \bot$ when $c$ has not been seen, or $c$ has been seen and followed by a \smaller\ priority, $(ii)$ or $l_c$ is the number of steps from $c$ to the current vertex $v$. If $l_c = \bot$ and a new occurrence of $c$ is seen, $l_c$ is initialized to $0$. If $l_c \neq \bot$, if a \smaller\ priority is seen in at most $\lambda-1$ steps, we reset $l_c$ to $\bot$, otherwise we move to a special vertex. Formally, we define $V' = V \times \{\bot,0,1,\ldots,\lambda - 2\}^{\frac{d}{2}} \cup \beta_V$, where the additional set $\beta_V =  \{\beta_v \mid v \in V \}$ is composed of the special vertices. 

Let $u = (v,\bar l) \in V' \setminus \beta_V$ be such that $v \in V$, $\bar l$ is a vector of counters $l_c$, one for each odd priority $c$. Given $(v,v') \in E$, we construct the edge $(u,u') \in E'$ such that 
\[
u' = \begin{cases}
\beta_{v'} & \mbox{if } \exists c,~ l_c = \lambda - 2 \mbox{ and } p(v') \not\stronger c,\\
(v',\bar l') & \mbox{otherwise},
\end{cases}
\]
where in the second case
\[
l'_c = \begin{cases}
l_c + 1 & \mbox{if } l_c \neq \bot \mbox{ and } p(v') \not\stronger c,\\
\bot      & \mbox{if } l_c \neq \bot \mbox{ and } p(v') \stronger c,\\
0          & \mbox{if } l_c = \bot \mbox{ and } p(v') = c,\\
\bot       & \mbox{otherwise}.
\end{cases}
\]
Notice that in the previous cases, if $l_c \neq \bot$ and $p(v') = c$, we do not define $l'_c = 0$. Indeed, if a new occurrence of $c$ is detected ($p(v') = c$), we have to check that the previous occurrence of $c$ ($l_c \neq \bot$) is followed by a \smaller\ priority $c'$ within at most $\lambda-1$ steps. If this happens, the last occurrence will automatically be followed by the same \smaller\ priority $c'$ and it is not necessary to keep track of the second occurence specifically.  

For all $v \in V$, we also add the edge $(\beta_v,(v,\bar l))$ to $E'$ such that $l_c = \bot$ for all odd priorities $c$ except if $p(v)$ is odd in which case $l_{p(v)} = 0$. In this way, we allow to delay the check performed on each odd priority. If $v_0$ is an initial vertex of $G$, then the corresponding initial vertex in $G'$ is $u_0 = (v_0,\bar l)$ such that $\bar l$ is defined as done previously. Finally we define the objective $\Obj' = \CoBuchi(U')$ with $U' = V' \setminus \beta_V$. One can check that $\playerOne$ has a winning strategy from $v_0$ in $(G,\Obj)$ if and only if $\playerOne$ has a winning strategy from $u_0$ in $(G',\Obj')$. As $|V'| = \mathcal{O}(|V| \cdot \lambda^{\frac{d}{2}})$ and $|E'| = \mathcal{O}(|E| \cdot \lambda^{\frac{d}{2}})$, the size of the game $G'$ is polynomial in the size of the original game $G$ since $\lambda < |V|$ and $d$ is fixed. As deciding the winner in the co-B\"uchi game $(G',\Obj')$ can be solved in $\mathcal{O}(|V'|^2)$ time~\cite{ChatterjeeH14}, deciding the winner in the game $(G,\Obj)$ can be solved in $\mathcal{O}( (|V|^2 \cdot \lambda^d)$ time. Moreover, as memoryless strategies are sufficient for both players to win in $(G',\Obj')$, finite-memory strategies with $\mathcal{O}(\lambda^{\frac{d}{2}})$ memory are sufficient for both players to win in $(G,\Obj)$. We show in Example~\ref{ex:memory} that memory is necessary for both players.

We now turn to the direct variant. The construction of the game $G'$ is rather similar, except that a unique absorbing special vertex $\beta$ is sufficient, and $\Obj' = \Safe(U')$ with $U' = V' \setminus \{\beta\}$. Then, $\playerOne$ wins the game $(G,\DirFWP(\lambda,p))$ if and only if he wins the new game $(G',\Safe(U'))$. The corresponding algorithm runs in $\mathcal{O}(|V'| + |E'|) = \mathcal{O}((|V| + |E|) \cdot \lambda^\frac{d}{2})$ time. Moreover, the memoryless winning strategies in the safety game $(G',\Obj')$ lead to finite-memory strategies with $\mathcal{O}(\lambda^\frac{d}{2})$ memory in the original game. The need for memory is also presented in Example~\ref{ex:memory}.

\medskip
$(c)$ We now proceed to the proof of Item~\ref{item:fixed3}. For the undirect variant, we construct from $G$ a game $G'$ that keeps in its vertices the last $\lambda$ vertices (of $G$) seen including the current vertex $v$. Formally, we define $V' = V^{\lambda}$. For each $u = (w_1, \ldots, w_{\lambda-1},v) \in V'$, and $(v,v') \in E$, we construct the edge $(u,u') \in E'$ such that $u' = (w_2,\ldots,w_{\lambda-1}, v, v')$.
Given $v_0$ an initial vertex in $G$, we let $u_0 = (w, \ldots, w,v_0)$ be the corresponding initial vertex of $G'$ such that $w$ is a vertex of $V$ with the highest even priority ($w$ is chosen in a way to have no influence for the objectives considered here). We also define the objective $\Obj' = \CoBuchi(U')$ with $U' = \{ (w_1, \ldots,w_{\lambda}) \in V' \mid \exists\, l \in \{0,\ldots,\lambda-1\} \mbox{ such that } p(w_{l+1}) \stronger p(w_1)\}$. Clearly, $\playerOne$ has a winning strategy from $v_0$ in $(G,\Obj)$ if and only if $\playerOne$ has a winning strategy from $u_0$ in $(G',\Obj')$. Note that the size of the game $G'$ is polynomial in the size of the original game $G$, with $|V'| = \mathcal{O}(|V|^{\lambda})$ and $|E'| = \mathcal{O}(|E| \cdot |V|^{\lambda-1})$ since $\lambda$ is fixed. Therefore deciding the winner in $(G,\Obj)$ can be done in $\mathcal{O}(|V|^{2  \lambda})$ time and both players have finite-memory winning strategies with $\mathcal{O}(|V|^{\lambda-1})$ memory. We show in Example~\ref{ex:memory} that both players need memory to win.

For the direct variant, the game $G'$ is identical but with the safety objective $\Obj' = \Safe(U')$. We get an algorithm in $\mathcal{O}((|V|+|E|)\cdot |V|^{\lambda-1})$ time and winning strategies with $\mathcal{O}(|V|^{\lambda-1})$ memory for both players. Example~\ref{ex:memory} shows that memory is necessary for both players.
\qed\end{proof}

The next example shows that both players need memory in fixed \textsf{PR} games with fixed parameters.

\begin{example} \label{ex:memory} 
Consider the game depicted in Figure~\ref{fig:memoryj1} where all vertices belong to $\playerOne$. We claim that $\playerOne$ needs memory to win for $({\sf Dir})\PR(\lambda,p)$ with $\lambda = 4$. Indeed, if he plays memoryless by always going to~$v_1$ from $v_0$, then in the resulting outcome $\pi = (v_0v_1v_2)^\omega$ the odd priority $1$ of vertex $v_2$ is followed by no \smaller\ priority, showing that $\pi \notin \PR(4,p)$ (and thus $\pi \notin \DirPR(4,p)$). Now, if $\playerOne$ always goes to $v_3$ from $v_0$ then the priority $3$ of vertex $v_5$ is followed by the \smaller\ priority $0$ of vertex $v_4$ in $4$ steps ($\lambda' = 5$), showing that $\pi = (v_0v_3v_4v_5v_6)^{\omega} \notin \PR(4,p)$ (hence $\pi \notin \DirPR(4,p)$). Thus $\playerOne$ has no memoryless winning strategy for $({\sf Dir})\PR(4,p)$. However, if he alternates between the two cycles, thus producing the outcome $\pi = (v_0v_1v_2v_0v_3v_4v_5v_6)^{\omega}$, one can check that $\pi$ belongs to $\DirPR(4,p)$, and thus also to $\PR(4,p)$.

\begin{figure}[ht]
\begin{minipage}[T]{.45\linewidth}
\centering
  \begin{tikzpicture}[scale=5]
    \everymath{\footnotesize}
	\draw (0,0) node [circle, draw] (A) {$v_0$};
	\draw (-0.5,0.25) node [circle, draw] (B) {$v_1$};
	\draw (0,0.25) node [circle, draw] (C) {$v_2$};

	\draw (0.08,-0.08) node [] (AA) {$3$};
	\draw (-0.45,0.16) node [] (BB) {$2$};
	\draw (0.08,0.17) node [] (CC) {$1$};	
	
	\draw (0.5,0.25) node [circle, draw] (D) {$v_3$};
	\draw (0.5,0) node [circle, draw] (E) {$v_4$};
	\draw (0.5,-0.25) node [circle, draw] (F) {$v_5$};
	\draw (0,-0.25) node [circle, draw] (G) {$v_6$};
	
	\draw (0.58,0.17) node [] (DD) {$1$};
	\draw (0.58,-0.08) node [] (EE) {$0$};
	\draw (0.58,-0.33) node [] (FF) {$3$};
	\draw (0.08,-0.33) node [] (GG) {$3$};

    \draw[->,>=latex] (A) to (B);
    \draw[->,>=latex] (B) to (C);
    \draw[->,>=latex] (C) to (A);
    \draw[->,>=latex] (A) to (D);
    \draw[->,>=latex] (D) to (E);
    \draw[->,>=latex] (E) to (F);
    \draw[->,>=latex] (F) to (G);
    \draw[->,>=latex] (G) to (A);

	\path (-0.15,0) edge [->,>=latex] (A);
	    \end{tikzpicture}
\caption{$\playerOne$ needs to alternate between the two simple cycles to win for $({\sf Dir})\PR(4,p)$.}
\label{fig:memoryj1}
\end{minipage}\hfill
\begin{minipage}[T]{.45\linewidth}
\centering
  \begin{tikzpicture}[scale=5]
    \everymath{\footnotesize}
    \draw (0,0) node [rectangle, inner sep = 5pt, draw] (A) {$v_0$};
    \draw (0.5,0.25) node [rectangle, inner sep = 5pt, draw] (B) {$v_1$};
    \draw (0,0.25) node [rectangle, inner sep = 5pt, draw] (C) {$v_2$};
    \draw (0.5,-0.25) node [rectangle, inner sep = 5pt, draw] (D) {$v_3$};
    \draw (0,-0.25) node [rectangle, inner sep = 5pt, draw] (E) {$v_4$};
    
    \draw (0.1,-0.1) node [] (AA) {$1$};
    \draw (0.6,0.15) node [] (BB) {$0$};
    \draw (0.1,0.15) node [] (CC) {$1$};
    \draw (0.6,-0.35) node [] (DD) {$1$};
    \draw (0.1,-0.35) node [] (EE) {$0$};
    
    \draw[->,>=latex] (A) to (B);
    \draw[->,>=latex] (B) to (C);
    \draw[->,>=latex] (C) to (A);
    \draw[->,>=latex] (A) to (D);
    \draw[->,>=latex] (D) to (E);
    \draw[->,>=latex] (E) to (A);
  
	\path (-0.15,0) edge [->,>=latex] (A);

  \end{tikzpicture}
  \vspace{0mm}
\caption{$\playerTwo$ needs to alternate between the two simple cycles to win for $\overline{({\sf Dir})\PR(3,p)}$.}
\label{fig:memoryj2}
\end{minipage}
\end{figure}

Let us now focus on the game depicted in Figure~\ref{fig:memoryj2}, where all vertices belong to $\playerTwo$. We show that $\playerTwo$ needs memory to win both $\overline{({\sf Dir})\PR(\lambda,p)}$ objectives with $\lambda = 3$. If $\playerTwo$ plays memoryless by always going to $v_1$ from $v_0$, then the resulting play $(v_0v_1v_2)^{\omega}$ belongs to $\DirPR(3,p) \subseteq \PR(3,p) $, and similarly if he always goes to $v_3$ from $v_0$. He is thus losing for $\overline{({\sf Dir})\PR(3,p)}$. However, if $\playerTwo$ alternates between the two cycles to produce the outcome $\pi = (v_0v_1v_2v_0v_3v_4)^{\omega}$, we have $\pi \not\in \PR(3,p)$ (hence $\pi \notin \DirPR(3,p)$) since priority $1$ of vertex $v_2$ is followed by the \smaller\ priority $0$ of vertex $v_4$  in $3$ steps ($\lambda' = 4$). This way, $\playerTwo$ wins for $\overline{({\sf Dir})\PR(3,p)}$.\hfill$\triangleleft$

\end{example}

\paragraph*{\bf Window parity objectives.}

Whereas deciding the winner in $({\sf Dir})\PR$ games is {\sf PSPACE}-complete, we establish in the next theorem that it is {\sf P}-complete for $({\sf Dir})\FWP$ games. 
First observe that if $\lambda \geq \frac{d}{2} \cdot |V|$, the problem boils down to solving the bounded variant thanks to Theorem~\ref{thm:finitaryparity}. Hence, we focus on the case where $\lambda < \frac{d}{2} \cdot |V|$. 

Our algorithm is inspired by the approach developed for \textit{window mean-payoff games} in~\cite{Chatterjee0RR15}. It can be sketched as follows. As for the fixed-parameter algorithms for $({\sf Dir})\PR$ games presented in Theorem~\ref{thm:fixed}, we want to reduce the $\FWP$ and $\DirFWP$ games to co-B\"uchi and safety games respectively, where $\playerOne$ wants to avoid ``bad vertices'' representing a violation of the condition at stake. Here, such a violation represents a $\lambda$-\textit{bad window}, i.e., a window for which no even minimum priority is found before $\lambda$ steps (see the terminology in Section~\ref{subsec:decomposition}). Detecting such $\lambda$-bad windows can be achieved by considering an extended game structure where we encode additional information for the minimum priority of the current window and the number of steps in this window. A ``bad vertex'' is visited whenever we reach the end of a $\lambda$-window with an odd minimum priority. If an even minimum is found, a \good\ history is detected and the step counter is reset. The extended game has size $\mathcal{O}(|V|\cdot d \cdot \lambda)$, hence polynomial size since $\lambda < \frac{d}{2} \cdot |V|$. Therefore, we can solve it in polynomial time. This is in contrast to window mean-payoff games where the fixed variant requires \textit{pseudo}-polynomial time in general~\cite{Chatterjee0RR15}.

Upper bounds on the memory are obtained by construction of our reduction and we prove polynomial lower bounds in the upcoming Example~\ref{ex:necmemory}.

\begin{theorem}\label{prop:WPonedim}
Let $G = (V_1,V_2,E)$ be a game structure, $v_0$ be an initial vertex, $p$ be a priority function, and $\Obj$ be the objective $\DirFWP(\lambda,p)$ (resp. $\FWP(\lambda,p)$) for some $\lambda \in \Nzero$. Then deciding the winner in $(G,\Obj)$ from $v_0$ is {\sf P}-complete.
\begin{enumerate}
\item If $\lambda \geq \frac{d}{2} \cdot |V|$, deciding the winner can be done in $\mathcal{O}(|V| \cdot |E|)$ (resp. $\mathcal{O}(|V|^2 \cdot |E|)$) time, memoryless strategies are sufficient for $\playerOne$ and linear-memory strategies are necessary and sufficient for $\playerTwo$ (resp. infinite memory is necessary for $\playerTwo$).\label{item:WPonedim1}
\item If $\lambda < \frac{d}{2} \cdot |V|$, deciding the winner can be done in $\mathcal{O}((|V| + |E|) \cdot d \cdot \lambda)$ (resp. $\mathcal{O}(|V|^2 \cdot d^2 \cdot \lambda^2)$) time, and polynomial-memory strategies with $\mathcal{O}(d \cdot \lambda)$ memory are sufficient for both players. Moreover, polynomial memory is necessary for both players.
\end{enumerate}
\end{theorem}

\begin{proof}
The reduction from reachability games used for Theorem~\ref{thm:finitaryparity} also suffices to obtain $\sf P$-hardness for $({\sf Dir})\FWP$ games, so it remains to establish a polynomial-time algorithm and study the memory requirements for the winning strategies.
If $\lambda \geq \frac{d}{2} \cdot |V|$, the results of Item~\ref{item:WPonedim1} follow from Theorem~\ref{thm:finitaryparity}. Hence we now suppose that $\lambda < \frac{d}{2} \cdot |V|$.

We begin by studying the \textit{undirect} variant.  By Lemma~\ref{lem:goodDec}, a play belongs to $\FWP(\lambda,p)$ if and only if it has an \EGoodDec{\lambda}. Therefore from $G$, we construct a game $G'$ able to detect \good\ histories. That is, we keep in the vertices of $G'$ the current vertex of $G$, the minimum priority of the current window and the number of steps performed in the current window. As soon as the minimum priority is even (in at most $\lambda-1$ steps), a \good\ history has been detected, and the information is reset with a new window. More precisely, we define $V' = V \times \{0,\ldots,d\} \times\{0,\ldots,\lambda - 1\} \cup \beta_V$, where the additional set $\beta_V =  \{\beta_v \mid v \in V \}$ is composed of special vertices for the detection of a $\lambda$-\bad\ window. Let $u = (v,c,l) \in V' \setminus \beta_V$ be such that $v \in V$, $c$ is the current minimum priority and $l$ is the current number of steps. Given $(v,v') \in E$, we construct the edge $(u,u') \in E'$ such that
\[
u' = \begin{cases}
(v',p(v'),0) &\mbox{ if } c \mbox{ is even}, \\
(v',\min(c,p(v')),l+1) &\mbox{ if } c \mbox{ is odd and } l < \lambda-1, \\
\beta_{v'} &\mbox{ otherwise}.
\end{cases}
\]
We also add the edges $(\beta_v,(v,p(v),0))$ to $E'$ for all $v \in V$. If $v_0$ is an initial vertex of $G$, then the corresponding initial vertex in $G'$ is $u_0 = (v_0,p(v_0),0)$. Finally, we define the objective $\Obj' = \CoBuchi(U')$ with $U' = V' \setminus \beta_V$. Observe that a play winning for $\Obj'$ corresponds to a play accepting an \textit{eventually \good\ decomposition} (which is easily obtained by looking at the ``resets'' of the step counter).

Clearly, thanks to Lemma~\ref{lem:goodDec}, $\playerOne$ has a winning strategy from $v_0$ in $(G,\Obj)$ if and only if $\playerOne$ has a winning strategy from $u_0$ in $(G',\Obj')$. As $|V'| = \mathcal{O}(|V| \cdot d \cdot \lambda)$ and $|E'| = \mathcal{O}(|E| \cdot d \cdot \lambda)$, the size of the game $G'$ is polynomial in the size of the original game $G$ (since $\lambda < \frac{d}{2} \cdot |V|$). It follows that deciding the winner in the game $(G,\Obj)$ can be solved in $\mathcal{O}(|V|^2 \cdot d^2 \cdot \lambda^2)$ time~\cite{ChatterjeeH14}, and that finite-memory strategies with $\mathcal{O}(d \cdot \lambda)$ memory are sufficient for both players to win.

Let us turn to the \textit{direct} variant. By Lemma~\ref{lem:goodDec}, a play belongs to $\DirFWP(\lambda,p)$ if and only if it has a \GoodDec{\lambda}. The construction of the game $G'$ is thus rather similar, except that a unique absorbing vertex $\beta$ is sufficient for the detection of a $\lambda$-\bad\ window, and $\Obj' = \Safe(U')$ with $U' = V' \setminus \{\beta\}$. Then, $\playerOne$ wins the game $(G,\DirFWP(\lambda,p))$ if and only if he wins the new game $(G',\Safe(U'))$. The corresponding algorithm runs in $\mathcal{O}((|V| + |E|) \cdot d \cdot \lambda)$ time. Moreover, both players have finite-memory strategies with $\mathcal{O}(d \cdot \lambda)$ memory. 

The necessity of polynomial memory is established in Example~\ref{ex:necmemory}.\qed
\end{proof}
 
In the next example, we illustrate the need for polynomial memory, for both players, in $({\sf Dir})\FWP$ games. We use game families proposed in~\cite{ChatterjeeF13}.

\begin{example} \label{ex:necmemory} Consider the game in Figure~\ref{fig:P1needsmemory} where the unlabeled vertices have all priority $d = 6$. This example can easily be generalized to any even $d \geq 0$, and observe that the size of the game is $|V| = 2 + \frac{d}{2} \cdot (\frac{d}{2} + 1)$, hence polynomial in $d$. 

\begin{figure}[htb]
\centering
  \begin{tikzpicture}[scale=5]
    \everymath{\footnotesize}
    \draw (0,0) node [rectangle, draw] (A) {$v_0$};
    \draw (0.3,0) node [circle, draw] (B) {};
    \draw (0.6,0) node [circle, draw] (C) {};
    \draw (1.2,0) node [circle, draw] (G) {};
    
	\draw (0.3,-.07) node[] (E) {$3$};    
    
    \draw[->,>=latex] (A) to (B);
    \draw[->,>=latex] (B) to (C);
    \draw[->,>=latex] (C) to (G);
    
    \draw (0.3,-0.2) node [circle, draw] (H) {};
    \draw (0.3,-.27) node[] (HH) {$1$};
    \draw[->,>=latex] (A) to (H);
    
    \draw (0.3,0.2) node [circle, draw] (I) {};
    \draw (0.6,0.2) node [circle, draw] (J) {};
    \draw (0.9,0.2) node [circle, draw] (K) {};
    \draw (0.3,0.13) node[] (II) {$5$};
    \draw[->,>=latex] (A) to (I);
    \draw[->,>=latex] (I) to (J);
    \draw[->,>=latex] (J) to (K);
    \draw[->,>=latex] (K) to (G);
    \draw[->,>=latex] (H) to (G);
    
    \draw (1.5,0.2) node [circle, draw] (L) {};
    \draw (1.5,0.13) node[] (LL) {$4$};
    \draw (1.5,0) node [circle, draw] (M) {};
    \draw (1.5,-0.2) node [circle, draw] (N) {};
    
    \draw[->,>=latex] (G) to (L);
    \draw[->,>=latex] (G) to (M);
    \draw[->,>=latex] (G) to (N);
    
    \draw (1.8,0) node [circle, draw] (O) {};
    \draw (1.8,-0.07) node[] (OO) {$2$};
    \draw (1.8,-0.2) node [circle, draw] (P) {};
    
    \draw[->,>=latex] (M) to (O);
    \draw[->,>=latex] (N) to (P);
    
    \draw (2.1,-0.2) node [circle, draw] (Q) {};
     \draw (2.1,-0.27) node[] (QQ) {$0$};
    
   \draw[->,>=latex] (P) to (Q);
  
    \draw (2.4,0) node [] (R) {back to $v_0$};
    \draw[->,>=latex] (Q) to (R);
    \draw[->,>=latex] (O) to (R);
    \draw[->,>=latex] (L) to (R);
    
	\path (-0.2,0) edge [->,>=latex] (A);    
    
    \end{tikzpicture}
    \caption{In order to win for objective $({\sf Dir})\FWP(5, p)$, $\playerOne$ needs to answer to priority $5$ (resp.~$3$, $1$) by choosing the path with priority $4$ (resp.~$2$, $0$). This requires $\mathcal{O}(\frac{d}{2})$ memory (here, $d = 6$).} 
\label{fig:P1needsmemory}
\end{figure}
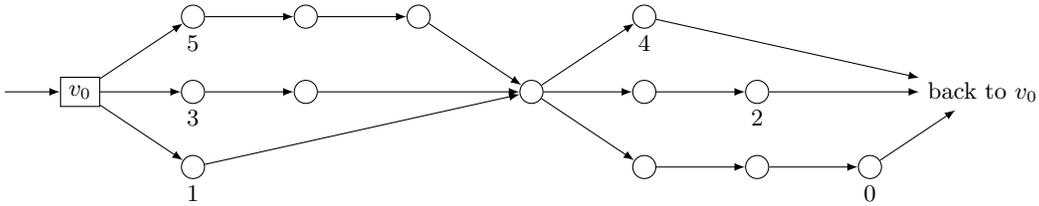

We claim that $\playerOne$ needs memory of size $\frac{d}{2}$ to win the game $(G,\Obj)$ with $\Obj = ({\sf Dir})\FWP(\lambda,p)$ with $\lambda = \frac{d}{2} + 2$. Indeed, each time $\playerTwo$ chooses the path with priority~$(d-1)$ (resp. $(d-3)$, \ldots{}, $1$), the only possibility for $\playerOne$ is to choose the path with priority $(d-2)$ (resp. $(d-4)$, \ldots{}, $0$) otherwise he creates a $\lambda$-bad window. If $\playerOne$ uses a finite-memory strategy with less than $\frac{d}{2}$ memory, he has to answer to two different odd priorities with the same choice, and $\playerTwo$ can take advantage of this to create $\lambda$-bad windows at every visit of~$v_0$. Hence, $\playerTwo$ can prevent $\playerOne$ from winning for $\FWP(\lambda,p)$, a fortiori for $\DirFWP(\lambda,p)$.

We thus have a family of games polynomial in $d$ and for which $\playerOne$ needs polynomial memory to win for objective $({\sf Dir})\FWP(\lambda,p)$.

We now turn to a second example to illustrate the necessity of polynomial memory for $\playerTwo$. Consider the family of game structures (parameterized by $n \geq 2$) depicted in Figure~\ref{fig:P2needsmemory}. From vertex $v_0$, $\playerTwo$ can choose one among the $n$ outgoing edges $(v_0,v_{\neq i})$, and from vertex $v_{\neq i}$, $\playerOne$ can choose one among the $(n-1)$ outgoing paths $\rho_j = v_1 \ldots v_n$ of length $n$, with $j \neq i$, the vertices of which all have priority $1$ except vertex~$v_j$ having priority~$0$. Vertices $v_0$ and $v_{\neq i}$, $i \in \{1, \ldots, n\}$, all have priority $1$. These choices of both players alternate infinitely many times. Observe that the size of this game is $|V| = 1 + n\cdot(n+1)$, hence polynomial in $n$.

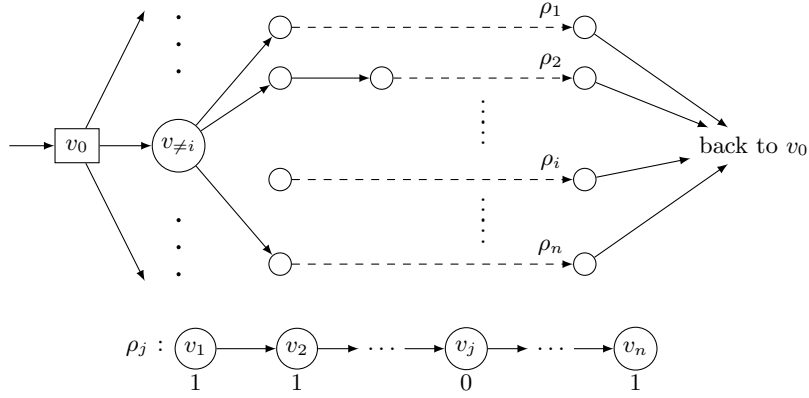
\begin{figure}[h]
\centering
  \begin{tikzpicture}[scale=4.5]
    \everymath{\footnotesize}
    \draw (0,0) node [rectangle, draw,inner sep=4pt] (A) {$v_0$};
	\path (-0.2,0) edge [->,>=latex] (A);    
	\path (A) edge [->,>=latex] (0.2,0.4);
	\path (A) edge [->,>=latex] (0.2,-0.4);
	\draw (0.3,0) node [circle,draw,inner sep=2pt] (B) {$v_{\ne i}$};
	\draw[->,>=latex] (A) to (B);
	\draw (0.3,0.38) node [circle,draw,inner sep=0.4pt,fill] {};
	\draw (0.3,0.3) node [circle,draw,inner sep=0.4pt,fill] {};
	\draw (0.3,0.22) node [circle,draw,inner sep=0.4pt,fill] {};
	\draw (0.3,-0.38) node [circle,draw,inner sep=0.4pt,fill] {};
	\draw (0.3,-0.30) node [circle,draw,inner sep=0.4pt,fill] {};
	\draw (0.3,-0.22) node [circle,draw,inner sep=0.4pt,fill] {};
	
	\draw (0.6,0.35) node [circle,draw,inner sep=3pt] (C) {};
	\draw (0.6,0.29) node [] {};
	\draw (0.6,0.2) node [circle,draw,inner sep=3pt] (D) {};
	\draw (0.6,-0.1) node [circle,draw,inner sep=3pt] (E) {};
	\draw (0.6,-0.35) node [circle,draw,inner sep=3pt] (F) {};
	\draw[->,>=latex] (B) to (C);
	\draw[->,>=latex] (B) to (D);
	\draw[->,>=latex] (B) to (F);
	\draw (0.9,0.2) node [circle,draw,inner sep=3pt] (DD) {};
	\draw (0.9,0.14) node [] {};
	\draw[->,>=latex] (D) to (DD);
	\draw (1.5,0.35) node [circle,draw,inner sep=3pt] (CC) {};
	\draw (1.5,0.2) node [circle,draw,inner sep=3pt] (DDD) {};
	\draw (1.5,-0.1) node [circle,draw,inner sep=3pt] (EE) {};
	\draw (1.5,-0.35) node [circle,draw,inner sep=3pt] (FF) {};
	\draw[->,>=latex,dashed] (DD) to (DDD);
	\draw[->,>=latex,dashed] (C) to (CC);
	\draw[->,>=latex,dashed] (E) to (EE);
	\draw[->,>=latex,dashed] (F) to (FF);
	\draw (1.4,0.4) node [] () {$\rho_1$};
	\draw (1.4,0.25) node [] () {$\rho_2$};
	\draw (1.4,-0.05) node [] () {$\rho_i$};
	\draw (1.4,-0.3) node [] () {$\rho_n$};
	
	\draw (1.2,0.1) node[] () {\rotatebox{90}{$\ldots$}};
	\draw (1.2,0.04) node[] () {\rotatebox{90}{$\ldots$}};
	\draw (1.2,-0.19) node[] () {\rotatebox{90}{$\ldots$}};
	\draw (1.2,-0.25) node[] () {\rotatebox{90}{$\ldots$}};
	
	\draw (2,0) node [] (G) {back to $v_0$};
	\draw[->,>=latex] (CC) to (G);
	\draw[->,>=latex] (DDD) to (G);
	\draw[->,>=latex] (EE) to (G);
	\draw[->,>=latex] (FF) to (G);
	
	\draw (0.2,-0.6) node [] () {$\rho_j$ :};
	\draw (0.35,-.6) node [circle,draw,inner sep=2pt] (M) {$v_1$};
	\draw (0.65,-.6) node [circle,draw,inner sep=2pt] (N) {$v_2$};
	\draw[->,>=latex] (M) to (N);
	\draw (0.9,-.6) node [] (OO) {$\ldots$};
	\draw[->,>=latex] (N) to (OO);
	\draw (1.15,-.6) node [circle,draw,inner sep=2pt] (O) {$v_j$};
	\draw[->,>=latex] (OO) to (O);
	\draw (1.4,-.6) node [] (OOO) {$\ldots$};
	\draw[->,>=latex] (O) to (OOO);
	\draw (1.65,-.6) node [circle,draw,inner sep=2pt] (P) {$v_{n}$};
	\draw[->,>=latex] (OOO) to (P);
	\draw (0.35,-.7) node [] {$1$};
	\draw (0.65,-.7) node []  {$1$};
	\draw (1.15,-.7) node []  {$0$};
	\draw (1.65,-.7) node []  {$1$};
	
    \end{tikzpicture}
    \caption{In order to win for objective $\overline{({\sf Dir})\FWP(n+3, p)}$, $\playerTwo$ needs to answer to the choice $\rho_i$ made by $\playerOne$ by taking vertex $v_{\neq i}$, hence eventually forcing an $(n+3)$-bad window when $\playerOne$ is forced to pick $\rho_j$ with $j > i$. This requires $\mathcal{O}(n)$ memory.} 
\label{fig:P2needsmemory}
\end{figure}

We claim that $\playerTwo$ needs memory of size $n$ to win the game $(G,\overline{\Obj})$ with $\Obj = {\sf (Dir)}\FWP(\lambda,p)$ such that $\lambda = n+3$. $(i)$ Let us first explain that he can win with such a memory. Indeed, at each alternation, $\playerTwo$ records the last choice $\rho_i$ of $\playerOne$ and then chooses $(v_0,v_{\neq i})$. At the next alternation, $\playerOne$ must choose $\rho_j$ with either $j > i$ or $j < i$. The first case necessarily occurs infinitely often since eventually $i$ equals $1$ if $\playerOne$ keeps choosing $j < i$. Now, observe that when $\playerOne$ chooses $j > i$, the play contains $(n+2)$ consecutive vertices with priority~$1$, hence a $\lambda$-bad window since $\lambda = n +3$. This shows that $\playerTwo$ wins for objective $\overline{\FWP(\lambda,p)}$ and thus for $\overline{\DirFWP(\lambda,p)}$ too. $(ii)$ Let us now show that $\playerTwo$ is losing with a memory of size less then $n$. In this case, there is a vertex $v_{\neq i}$ that he will never choose. Thus $\playerOne$ can win for objective $\DirFWP(\lambda,p)$ (a fortiori for $\FWP(\lambda,p)$) by choosing the path $\rho_i$ at each alternation. Indeed, this will only induce sequences of $(n+1)$ consecutive priorities equal to $1$ separated by priority $0$, which is fine since $\lambda = n +3$.

We thus have a family of games polynomial in $n$ and for which $\playerTwo$ needs polynomial memory to win for objective $\overline{({\sf Dir})\FWP(\lambda,p)}$.\hfill$\triangleleft$
\end{example}

\section{Multi-dimension games} \label{sec:manydim}

We now consider multi-dimension games: in this setting, there are $n$ priority functions $p_1$, \ldots{}, $p_n$ and the objective $\Obj$ is the \textit{conjunction} of \textit{identical} objectives $\Obj_m$ for each ``dimension'' (i.e., priority function), with  $\Obj_m$ being $({\sf Dir}){\sf FixX}$ or $({\sf Dir}){\sf BndX}$ for ${\sf X} \in \{{\sf PR}, {\sf WP}\}$. As in the one-dimension case, we first address the \textit{bounded} variants in Section~\ref{subsec:multibounded}, then turn to the \textit{fixed} ones in Section~\ref{subsec:multifixed}.

\subsection{Bounded variants} \label{subsec:multibounded}

Recall that Proposition~\ref{prop:inclusions} established the equality of objectives $({\sf Dir})\BWP(p)$ and $({\sf Dir})\BP(p)$ in the one-dimension setting. This equality trivially carries over to the multi-dimension setting, i.e., we have that $ \cap_{m=1}^n ({\sf Dir})\BWP(p_m) = \cap_{m=1}^n ({\sf Dir})\BP(p_m)$ since the individual objectives (one per priority function) are equal. Hence, it suffices to obtain our results for either \textsf{WP} or \textsf{PR} objectives.

\paragraph{\bf Overview.} The next theorem presents an overview of our results. For the sake of readability, its proof is split in several lemmas: we prove ${\sf EXPTIME}$-membership and upper bounds on memory in Lemma~\ref{lem:finitarystreett}, the equalities of Items~\ref{item:multidimequalities} and~\ref{item:multidimequalities2} in Lemma~\ref{lem:multibornes}, ${\sf EXPTIME}$-hardness in Lemma~\ref{prop:exptimehard}, and finally, lower bounds on memory in Lemma~\ref{lem:memexp}.

\begin{theorem} \label{prop:multiDirBP} 
Let $G = (V_1,V_2,E)$ be a game structure, $v_0$ be an initial vertex, $p_1,\ldots,p_n$ be $n$ priority functions, and $\Obj$ be the objective $\cap_{m=1}^n \DirBP(p_m)$ or $\cap_{m=1}^n \DirBWP(p_m)$ (resp. $\cap_{m=1}^n \BP(p_m)$ or $\cap_{m=1}^n \BWP(p_m)$). Let $b = |V| \cdot 2^{n\cdot\frac{d}{2}} \cdot n\cdot\frac{d}{2}$.
\begin{enumerate}
\item Deciding the winner in $(G,\Obj)$ from $v_0$ is $\sf EXPTIME$-complete with an algorithm in $\mathcal{O}(b^2)$ (resp. $\mathcal{O}(|V| \cdot b^2)$) time, and exponential-memory strategies are necessary and sufficient for both players (resp. for $\playerOne$ and infinite-memory is necessary for $\playerTwo$).
\vspace{1mm}
\item $\forall\, \lambda \geq b$, $\forall\, \lambda' \geq b \cdot \frac{d}{2}$, the winning sets for the following objectives are all equal: $\cap_{m=1}^n \BP(p_m)$, $\cap_{m=1}^n \PR(\lambda,p_m)$, $\cap_{m=1}^n \BWP(p_m)$, and $\cap_{m=1}^n \FWP(\lambda',p_m)$.\label{item:multidimequalities}
\vspace{1mm}
\item The equalities given in Item~\ref{item:multidimequalities} also hold for the direct variants ${\sf (Dir)}$. \label{item:multidimequalities2}
\end{enumerate}
\end{theorem}

Before continuing with the proof of these results, let us comment them. First, observe that multi-dimension bounded \textsf{WP} or \textsf{PR} games are $\sf EXPTIME$-complete whereas multi-dimension parity games are $\sf coNP$-complete~\cite{DBLP:conf/fossacs/ChatterjeeHP07}. Hence, in this case, the boundedness requirements specified by \textsf{WP} or \textsf{PR} objectives induce higher complexity, unlike in one-dimension games, where they yield a lower one ($\sf P$-complete instead of the long-standing ${\sf UP} \cap {\sf coUP}$ barrier of parity games). This implies that, in multi-dimension games, bounded \textsf{WP} or \textsf{PR} objectives cannot be used to approximate efficiently parity objectives, in contrast to one-dimension games. A similar dichotomy was already witnessed for \textit{window mean-payoff games} with regard to approximation of classical \textit{mean-payoff games}~\cite{Chatterjee0RR15}. 

Interestingly, the decidability of multi-dimension bounded window mean-payoff games is still open and they are known to be non-primitive-recursive-hard~\cite{Chatterjee0RR15}, whereas we prove here $\sf EXPTIME$-completeness for the parity counterpart of this objective. 
This suggests that the colossal complexity of bounded window mean-payoff games is a result of the quantitative nature of mean-payoff mixed with windows, and not an inherent drawback of the window mechanism.

\paragraph{\bf Exponential-time algorithm and upper bounds on memory.} To prove ${\sf EXPTIME}$-membership, we have to introduce related games from the literature. First, let us consider \textit{request-response games}~\cite{WallmeierHT03,ChatterjeeHH11}. Consider $r$ sets of vertices $Rq_1,\ldots, Rq_r$ representing requests and $r$ sets of vertices $Rp_1,\ldots,Rp_r$ representing the corresponding responses ($Rq_i, Rp_i \subseteq V$ for all $i$). 
The \emph{request-response}
objective $\RR((Rq_i,Rp_i)_{i=1}^r)$ requires that \textit{for all} $i$, whenever a vertex of $Rq_i$ is visited, then, later on, a vertex of $Rp_i$ is also visited.\footnote{Note that a single response $Rp_i$ suffices to answer all pending requests $Rq_i$, in the same spirit as for priorities in the \textit{parity-response} objective.} Observe that by definition, this objective is \textit{direct}, i.e., the condition must hold from the start, not only eventually.
Solving these games is ${\sf EXPTIME}$-complete with an algorithm in $\mathcal{O}((|V| \cdot 2^r \cdot r)^2)$ time, and exponential memory is both sufficient and necessary for both players \cite{WallmeierHT03,ChatterjeeHH11}.

In \cite{ChatterjeeHH09}, Chatterjee et al.~studied \textit{bounded Streett} games, which using our terminology for the sake of consistency, can be equivalently seen as \emph{direct bounded request-response} games. The corresponding objective $\DirBndRR((Rq_i,Rp_i)_{i=1}^r)$ asks that there exist a bound $b \in \mathbb{N}_0$ such that if a request is visited, then the corresponding response is visited within $b$ steps. Chatterjee et al.~proved that such a bound always exists when the (``unbounded'') request-response objective $\RR$ can be won by $\playerOne$, as a by-product of the construction used to solve these games in~\cite{WallmeierHT03}.
Ergo, $\RR$ games are equivalent to $\DirBndRR$ games.

A prefix-independent variant of the $\DirBndRR$ objective is also studied in~\cite{ChatterjeeHH09}, under the name of \textit{finitary Streett}. Again, to maintain consistency, we call it the \emph{bounded request-response objective} $\BndRR$. It is naturally defined from the direct variant $\DirBndRR$ in the same way as all \textit{undirect} variants in this paper (Definition~\ref{def:obj}).

We sum up the results\footnote{The results presented here are based on the best known complexity in $\mathcal{O}(|V|^2)$ for solving B\"uchi games~\cite{ChatterjeeH14}, and not on the previously best known complexity in $\mathcal{O}(|V|\cdot |E|)$ originally used in \cite{WallmeierHT03,ChatterjeeHH09}.} for all these games in the next theorem.

\begin{theorem}[\cite{WallmeierHT03,ChatterjeeHH09,ChatterjeeHH11}] \label{thm:finitarystreett}
Let $G = \Gdev$ be a game structure, $v_0$ be an initial vertex, and $(Rq_i,Rp_i)_{i=1}^r$ be a set of $r$ pairs of requests and responses, such that $Rq_i, Rp_i \subseteq V$. Let $b = |V| \cdot 2^r \cdot r$.
\begin{enumerate}
\item If $\playerOne$ has a winning strategy for $\DirBndRR((Rq_i,Rp_i)_{i=1}^r))$ (resp.~$\BndRR((Rq_i,Rp_i)_{i=1}^r))$), then he has one that enforces that (resp.~eventually) every request is followed by a corresponding response in at most $b$ steps.\label{item:finitarystreett1}
\item Deciding the winner in the game $(G,\DirBndRR((Rq_i,Rp_i)_{i=1}^r))$ from $v_0$ is ${\sf EXPTIME}$-complete with an algorithm in $\mathcal{O}(b^2)$ time, and exponential-memory strategies are sufficient for both players and necessary for $\playerOne$.
\item Deciding the winner in the game $(G,\BndRR((Rq_i,Rp_i)_{i=1}^r))$ from $v_0$ is in ${\sf EXPTIME}$ with an algorithm in $\mathcal{O}(|V| \cdot b^2)$ time, exponential-memory strategies are both sufficient and necessary for $\playerOne$, and infinite-memory is necessary for $\playerTwo$.
\end{enumerate}
\end{theorem}

We will now establish a polynomial-time reduction from multi-dimension $\DirBWP$ and $\BWP$ games (or equivalently, $\DirBP$ and $\BP$ games) to $\DirBndRR$ and $\BndRR$ games respectively. The crux is to consider $n\cdot\frac{d}{2}$ pairs of requests and responses, so that a request is made when an odd priority occurs and the corresponding response is the occurrence of a \smaller\ priority.
Thanks to Theorem~\ref{thm:finitarystreett}, we thus obtain an ${\sf EXPTIME}$ algorithm for multi-dimension ${\sf (Dir)}\BWP$ and ${\sf (Dir)}\BP$ games.

\begin{lemma}\label{lem:finitarystreett}
Let $G = (V_1, V_2, E)$ be a game structure, $v_0$ be an initial vertex, $p_1,\ldots,p_n$ be $n$ priority functions. Let $b = |V| \cdot 2^{n\cdot\frac{d}{2}} \cdot n\cdot\frac{d}{2}$.
\begin{enumerate}
\item Let $\Obj$ be the objective $\cap_{m=1}^n \DirBP(p_m)$ or $\cap_{m=1}^n \DirBWP(p_m)$. Deciding the winner from $v_0$ in the game $(G,\Obj)$ can be done in $\mathcal{O}(b^2)$ time and exponential-memory strategies are sufficient for both players.
\item Let $\Obj$ be the objective $\cap_{m=1}^n \BP(p_m)$ or $\cap_{m=1}^n \BWP(p_m)$. Deciding the winner from $v_0$ in the game $(G,\Obj)$ can be done in $\mathcal{O}(|V| \cdot b^2)$ time, exponential-memory strategies are sufficient for $\playerOne$ and infinite-memory is necessary for $\playerTwo$.
\end{enumerate}
\end{lemma}

\begin{proof}
In this proof, we only consider the objective $\Obj = \cap_{m=1}^n \DirBP(p_m)$ (resp. $\cap_{m=1}^n \BP(p_m)$) as $\cap_{m=1}^n ({\sf Dir})\BP(p_m) = \cap_{m=1}^n ({\sf Dir})\BWP(p_m)$ by Proposition~\ref{prop:inclusions} Item~\ref{item:samebnd}.

We prove this lemma by encoding the objective $\Obj$ as a $\DirBndRR$ (resp. $\BndRR$) objective. Intuitively, for any dimension, one wants to response to an odd priority by a \smaller\ priority. Formally, we define for each dimension $m$ and each \textit{odd} priority $c$ the following sets of vertices: $Rq_{m,c} = \{ v \in V \mid p_m(v) = c\}$ and $Rp_{m,c} = \{ v \in V \mid p_m(v) \strongerm c\} = \{v \in V \mid p_m(v) \in \{0,2,\ldots,c-1\} \}$. 
Clearly, $\playerOne$ has a winning strategy for the objective $\Obj$ from $v_0$ if and only if $\playerOne$ has a winning strategy from $v_0$ for the $\DirBndRR$ (resp. $\BndRR$) objective considering the $n\cdot\frac{d}{2}$ pairs $(Rq_{m,c},Rp_{m,c})$. Applying Theorem~\ref{thm:finitarystreett} with $r = n\cdot\frac{d}{2}$ we get the complexity and memory results stated in Lemma~\ref{lem:finitarystreett}, except the necessity of infinite-memory for $\playerTwo$. Recall that the latter property already holds in one-dimension $\BP$ and $\BWP$ games (see Theorem~\ref{thm:finitaryparity}).
\qed
\end{proof}

\paragraph{\bf Equalities between objectives.} The next lemma gives the last two items of Theorem~\ref{prop:multiDirBP}. The key ingredient is the bound given in Theorem~\ref{thm:finitarystreett} for $({\sf Dir})\BndRR$ games, and by extension, for $({\sf Dir})\BP$ and $({\sf Dir})\BWP$ games thanks to the reduction established in Lemma~\ref{lem:finitarystreett}. The rest follows the same lines as in the one-dimension case, i.e., it builds upon the inclusions and equalities presented in Proposition~\ref{prop:inclusions}.

\begin{lemma}\label{lem:multibornes}
Let $b = |V| \cdot 2^{n\cdot\frac{d}{2}} \cdot n\cdot\frac{d}{2}$. For all $\lambda \geq b$, $\lambda' \geq b\cdot\frac{d}{2}$, the winning sets for the following objectives are all equal: $\cap_{m=1}^n \BP(p_m)$, $\cap_{m=1}^n \PR(\lambda,p_m)$, $\cap_{m=1}^n \BWP(p_m)$, and $\cap_{m=1}^n \FWP(\lambda',p_m)$.
The same equalities also hold for the direct variants ${\sf (Dir)}$.
\end{lemma}
\begin{proof}
By Item~\ref{item:finitarystreett1} of Theorem~\ref{thm:finitarystreett} and the reduction established in Lemma~\ref{lem:finitarystreett}, the winning set for the objective $\cap_{m=1}^n \BP(p_m)$ is equal to the winning set for the objective $\cap_{m=1}^n \PR(\lambda,p_m)$ with $\lambda = |V| \cdot 2^{n\frac{d}{2}} \cdot n\cdot\frac{d}{2}$. Now, as in the one dimension case, the other equalities follow from Proposition~\ref{prop:inclusions} (Items~\ref{item:monotone},~\ref{item:WPcPR},~\ref{item:encadrement},~\ref{item:samebnd}). We get the equalities for the direct variant with the same proof.
\qed
\end{proof}

\paragraph{\bf Lower bound on complexity.} To prove the $\sf EXPTIME$-hardness of objective $({\sf Dir})\BWP$ (and equivalently, of objective $({\sf Dir})\BP$), we establish a reduction from the \textit{membership problem for alternating polynomial-space Turing machines (APTMs)}~\cite{DBLP:journals/jacm/ChandraKS81}. Our proof is adapted from the reduction presented in~\cite[Lemma 23]{Chatterjee0RR15} in the related context of window mean-payoff games. Similar techniques have been used for request-response games~\cite{ChatterjeeHH11}. Since technical details are similar to~\cite[Lemma 23]{Chatterjee0RR15}, we only include here a high-level sketch of the reduction. The main change is the way we open and close windows: whereas weights were used for window mean-payoff games, we need here to emulate the same actions with adapted priorities.
Interestingly, our proof also shows $\sf EXPTIME$-hardness of the fixed variants, $({\sf Dir})\FWP$ and $({\sf Dir})\PR$. Furthermore, the hardness already holds with only three priorities ($d = 2$).

\begin{lemma}\label{prop:exptimehard}
Let $G = (V_1,V_2,E)$ be a game structure, $v_0$ be an initial vertex, $p_1,\ldots,p_n$ be $n$ priority functions. Let $\Obj$ be an objective $\Obj = \cap_{m=1}^n \Obj_m$ such that $\forall\, m$, $\Obj_m = {\sf (Dir)BndPR}(p_m)$ (resp. ${\sf (Dir)BndWP}(p_m)$, ${\sf (Dir)FixWP}(\lambda,p_m)$, ${\sf (Dir)FixPR}(\lambda,p_m)$). Deciding the winner from $v_0$ in the game $(G,\Obj)$ is $\sf EXPTIME$-hard even if for all $m \in \{1,\ldots,n\}$, $p_m \colon V \rightarrow \{0,1,2\}$.
\end{lemma}

\begin{proof}[Sketch]
Given an APTM $\mathcal{M}$ and a word $\zeta \in \{0,1\}^*$, such that the tape contains at most $f(|\zeta|)$ cells, where $f$ is a polynomial function, the membership problem asks to decide if $\mathcal{M}$ accepts $\zeta$. It is well-known to be $\sf EXPTIME$-hard~\cite{DBLP:journals/jacm/ChandraKS81}. We first show how to reduce this problem to deciding if $\playerOne$ has a winning strategy in a game $G$ with an objective $\Obj = \cap_{m=1}^n \BWP(p_m)$. We discuss the other objectives later as the corresponding results are easily obtained with the same skeleton.

We build the game $G$ so that $\playerOne$ has to simulate the run of $\mathcal{M}$ on $\zeta$, and $\playerOne$ has a winning strategy in~$G$ if and only if the word is accepted by the machine. For each tape cell $h \in \{1, 2, \ldots, f(|\zeta|)\}$, we have two dimensions, $(h, 0)$ and $(h, 1)$. The game starts in a vertex $q_\text{in}$ that initializes those dimensions in order to encode the contents of the APTM tape: if the cell $h$ contains $1$ (resp.~$0$), then vertex $q_\text{in}$ has priority $1$ on dimension $(h,1)$ (resp.~$0$) and priority $0$ on dimension $(h,0)$ (resp.~$1$). In terms of windows, this means that when we start the game, we open a window in the dimensions that correspond to the actual contents of the tape. The goal for $\playerOne$ is now to correctly simulate the operation of the APTM by disclosing the correct symbols at each step.

The gadget used to simulate one step of the APTM is presented in Figure~\ref{fig:multiDimFixedMembership}. When $\playerOne$ reaches the vertex $(q,h)$, he must disclose the symbol under the tape head: he can either claim that it contains a $0$ and go to $(q, h, 0)_\text{check}$, which has priority $0$ in dimension $(h,0)$, or claim that the cell contains $1$ and go to $(q, h, 1)_\text{check}$, which has priority $0$ in dimension $(h,1)$. Priorities in all other dimensions are always set to $2$ (hence they do not open nor close any window). Intuitively, disclosing the correct symbol permits to close the open window on dimension $(h,i)$ whereas lying about the symbol leaves this window open.

\begin{figure}[htb]
  \centering   
  \scalebox{1}{\begin{tikzpicture}[->,>=stealth',shorten >=1pt,auto,node
    distance=2.5cm,bend angle=45,scale=0.6, font=\scriptsize]
    \tikzstyle{p1}=[draw,circle,text centered,minimum size=12mm]
    \tikzstyle{p2}=[draw,rectangle,text centered,minimum size=8mm]
    \tikzstyle{p3}=[draw, rounded rectangle, dashed, text centered,minimum size=8mm]
    \node[p1]  (0)  at (0, 0) {$(q, h)$};
    \node[p2]  (1) at (3, 2.5) {$(q, h, 0)_{\text{check}}$};
    \node[p2]  (2) at (3, -2.5)  {$(q, h, 1)_{\text{check}}$};
    \node[p2]  (3) at (6,0)  {$(q, h)_{\text{branch}}$};
    \node[p2]  (4)  at (10, 0) {$q_{\text{restart}}$};
    \node[p3]  (5)  at (8, 2.5) {$(q, h, 0)$};
    \node[p3]  (6)  at (8, -2.5) {$(q, h, 1)$};
    \path
    (0) edge (1)
    (0) edge (2)
    (1) edge (3)
    (1) edge (5)
    (2) edge (3)
    (2) edge (6)
    (3) edge (4)
    (5) edge (10, 2.5)
    (5) edge (10, 3.5)
    (5) edge (10, 1.5)
    (6) edge (10, -2.5)
    (6) edge (10, -3.5)
    (6) edge (10, -1.5);
    \draw[-,decorate,decoration={brace,amplitude=5pt},xshift=-3mm,yshift=0pt] (10.5,3.7) -- (10.5,1.3) node [black,midway,xshift=9pt] {Transitions of $(q, 0)$};
    \draw[-,decorate,decoration={brace,amplitude=5pt},xshift=-3mm,yshift=0pt] (10.5,-1.3) -- (10.5,-3.7) node [black,midway,xshift=9pt] {Transitions of $(q, 1)$};
      \end{tikzpicture}}
      \caption[Gadget ensuring a correct simulation of the APTM on tape cell $h$]{Gadget ensuring a correct simulation of the APTM on tape cell $h$.}
\label{fig:multiDimFixedMembership}
  \end{figure}
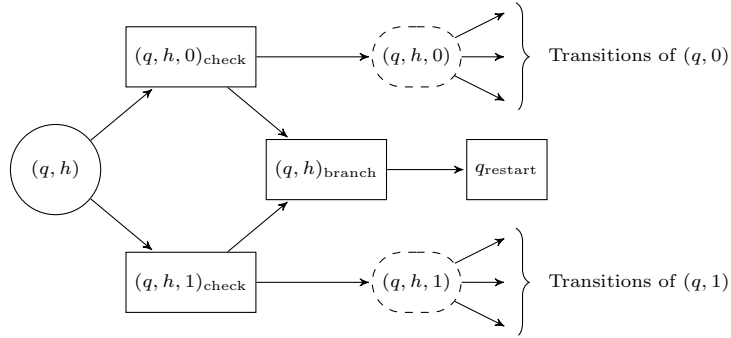

The job of $\playerTwo$ is to ensure that $\playerOne$ was faithful: in vertex $(q, h, i)_{\text{check}}$, $\playerTwo$ has the choice to either let the simulation continue by going to the vertex $(q, h, i)$ (where transitions of the APTM are encoded), or to go to $(q, h)_{\text{branch}}$ where the priority is $0$ in all dimensions except for $(h,0)$ and $(h,1)$. This means that if $\playerTwo$ decides to branch, all windows will be closed if $\playerOne$ was faithful, whereas an open window will remain on dimension $(h,i)$ if $\playerOne$ lied.

After $(q, h)_{\text{branch}}$, the game goes to $q_{\text{restart}}$, a vertex where the priority is $2$ in all dimensions and that contains a self-loop: this implies that if $\playerOne$ lied, $\playerTwo$ can take profit of this vertex to create an arbitrarily large open window. From $q_{\text{restart}}$, $\playerTwo$ additionally has the possibility to restart the game by going back to $q_{\text{in}}$ where the contents of the tape are encoded again for a new round of simulation.

Finally, the game also contains a vertex $q_{\text{acc}}$ that represents the accepting state of the APTM. To force $\playerOne$ to go toward this vertex, we use an additional dimension that sees priority $1$ at the beginning of the game (thus opening a window) and for which $q_{\text{acc}}$ and $(q, h)_{\text{branch}}$ are the only vertices with priority $0$ (hence the only ones able to close the open window). After visiting $q_{\text{acc}}$, the game goes to 
$q_{\text{restart}}$. Observe that the total number of dimensions used in our game is $(2\cdot f(|\zeta|) + 1)$.

Recall that we consider the objective $\Obj = \cap_{m=1}^n \BWP(p_m)$. Let $|\mathcal{C}|$ represent the size of the configuration graph of the APTM $\mathcal{M}$. It is known that $|\mathcal{C}|$ is a bound on the length of an accepting run in $\mathcal{M}$. In our game, this induces that if $q_{\text{acc}}$ can be reached, it can be reached in strictly less than $\lambda = 3\cdot |\mathcal{C}| + 3$ steps. We claim that $\playerOne$ has a winning strategy in our game if and only if the APTM $\mathcal{M}$ accepts the word $\zeta$.

First, assume that $\zeta$ is accepted by $\mathcal{M}$. Then, $\playerOne$ has a winning strategy consisting in always revealing the correct symbol. Indeed, either $\playerTwo$ never branches to a vertex $(q, h)_{\text{branch}}$ and the play reaches $q_{\text{acc}}$ and then $q_{\text{restart}}$ with all windows closed in strictly less than $\lambda$ steps; or $\playerTwo$ decides to branch to $(q, h)_{\text{branch}}$ but since $\playerOne$ never lied, all windows are closed when the play reaches $q_{\text{restart}}$ (again in strictly less than $\lambda$ steps). In both cases, the game can continue for another round of simulation in the same manner, yielding a play that is necessarily winning for $\Obj$ since it accepts a $\lambda$-good decomposition in each dimension.

Second, assume that $\zeta$ is not accepted by $\mathcal{M}$. If $\playerOne$ never cheats, then $\playerTwo$ can use the strategy that emulates the path in the run tree of $\mathcal{M}$ that never reaches the accepting state (this path exists since $\zeta$ is not accepted). Thus, the window in the last dimension can be kept open for an arbitrarily long time in order to make the round losing for $\playerOne$. Following this, $\playerTwo$ can restart the game and repeat this strategy. Note that in order to prevent $\playerOne$ from winning for objective $\Obj = \cap_{m=1}^n \BWP(p_m)$, it is necessary for $\playerTwo$ to keep restarting continuously and to increase the size of the open window at each round. The reason is essentially identical to the one in Example~\ref{ex:2}: ($i$) $\playerTwo$ must repeat losing rounds forever in order to counteract the prefix-independence of $\Obj$, ($ii$) rounds must be of increasing length to prevent the existence of a bound on the size of the windows.
Now, if $\playerOne$ does cheat, $\playerTwo$ can branch to $(q, h)_{\text{branch}}$ and reach $q_{\text{restart}}$ with an open window, here again creating an arbitrarily long open window, resulting in a losing round. Again, this can be repeated forever while increasing the size of open windows. Hence in both cases, $\playerOne$ has no winning strategy in $G$, which proves the correctness of our reduction.

It remains to discuss the adaptations needed for the other objectives. First observe that the right-to-left direction of the equivalence (i.e., case $\zeta$ not accepted) holds for any other objective $X$ claimed in Lemma~\ref{prop:exptimehard} since $X \subseteq \cap_{m=1}^n \BWP(p_m)$ by the various inclusions of Proposition~\ref{prop:inclusions}. For the other direction (i.e., case $\zeta$ accepted), observe that the strategy of $\playerOne$ that consists in always disclosing the correct symbol ensures that all windows close within $(\lambda -1)$ steps (with the $\lambda$ defined above): hence this strategy is also winning for $\cap_{m=1}^n \FWP(\lambda,p_m)$. The hardness also holds for the direct variants $\DirBWP$ and $\DirFWP$ as this strategy never produces any $\lambda$-bad window. Finally, the inclusions of Proposition~\ref{prop:inclusions} suffice to see that this strategy is also winning for all \textsf{PR} variants (with the same $\lambda$ for the fixed variants). Hence the proof holds for all claimed objectives.\qed
\end{proof}

\paragraph{\bf Lower bounds on memory.} The last missing pieces to the proof of Theorem~\ref{prop:multiDirBP} are the exponential lower bounds on memory. Recall that for $\playerTwo$ in \textit{undirect} bounded \textsf{WP} or \textsf{PR} games, we already proved that infinite memory is necessary in Example~\ref{ex:2}. The next lemma covers all remaining cases and establishes lower bounds matching the upper bounds granted by Lemma~\ref{lem:finitarystreett}. To achieve this, we prove a polynomial-time reduction from \textit{generalized reachability games}~\cite{FijalkowH13} to multi-dimension $({\sf Dir})\BWP$ and $({\sf Dir})\BP$ games. Since the former games are known to require exponential memory for both players, the reduction yields the desired lower bounds. A similar reduction is presented for window mean-payoff games in~\cite{Chatterjee0RR15}. Interestingly, the same technique also works for multi-dimension $({\sf Dir})\FWP$ and $({\sf Dir})\PR$ games.

Let us sketch the reduction from $\GenReach$ to multi-dimension $\FWP$ games (the other cases are similar). Intuitively, if the generalized reachability objective asks to visits $n$ different target sets, we will use $n$ dimensions. We create a modified version of the game structure such that, at the start of the game, we see priority $1$ in all dimensions, hence opening a window, and such that the only way to close the window in dimension $m \in \{1, \ldots{}, n\}$ is to visit the $m$-th target set. We also modify the game by giving $\playerTwo$ the possibility to close all open windows and restart the game by opening new ones: this is necessary to ensure that the prefix-independence of objective $\FWP$ cannot help $\playerOne$ to win without visiting all target sets. Finally, we use the fact that if $\playerOne$ has a winning strategy in a $\GenReach$ game with $n$ targets, then he has one that wins in strictly less than $n \cdot |V|$ steps (i.e., edges), to define an appropriate window size $\lambda = 2\cdot n \cdot |V|$ for which the reduction to objective $\cap_{m=1}^n \FWP(\lambda,p_m)$ on our modified game structure holds. As the reduction in Lemma~\ref{prop:exptimehard}, we only need three priorities here ($d=2$).

\begin{lemma}\label{lem:memexp}
Both players need exponential memory to win in games $(G,\Obj)$ where $\Obj = \cap_{m=1}^n \Obj_m$ such that $\forall\, m$, $\Obj_m = {\sf (Dir)BndPR}(p_m)$ (resp. ${\sf (Dir)BndWP}(p_m)$, ${\sf (Dir)FixWP}(\lambda,p_m)$, ${\sf (Dir)FixPR}(\lambda,p_m)$) even if for all $m \in \{1,\ldots,n\}$, $p_m \colon V \rightarrow \{0,1,2\}$.  
\end{lemma}

\begin{proof}
We use a polynomial reduction from generalized reachability games for which it is known that both players need exponential memory to win~\cite{FijalkowH13}. We begin with the case $\Obj = \cap_{m=1}^n \FWP(\lambda,p_m)$ and discuss the other cases at the end of the proof. 

Let $G^r = (V_1^r,V_2^r,E^r)$ be a game structure, $U_1,\ldots,U_n \subseteq V^r$ be target sets, and $v_0$ be the initial vertex. Let us recall that in the game $(G^r,\Obj^r)$ where $\Obj^r = \GenReach(U_1,\ldots,U_n)$, if $\playerOne$ has a winning strategy from~$v_0$ then there exists one which ensures visiting all sets $U_m$ in strictly less than $n \cdot |V^r|$ steps \cite{FijalkowH13}.

We build a game structure $G = \Gdev$ as follows. We define $V = V_1^r$ and $V_2 = V_2^r \cup V_{\text{branch}} \cup \{v_{\text{restart}}\}$ such that $V_{\text{branch}} = \{ b_{v,v'} \mid (v,v') \in E^r \}$ and $v_{\text{restart}}$ is a new vertex that is the initial vertex in $G$. We define $E$ as the set of edges such that $(v,b_{v,v'}), (b_{v,v'},v'), (b_{v,v'},v_{\text{restart}}) \in E$ for all $(v,v') \in E^r$, and $(v_{\text{restart}},v_0) \in E$. That is, we split each edge of $E^r$ into two edges such that $\playerTwo$ can decide in the central vertex $b_{v,v'}$ to branch to $v_{\text{restart}}$ or to continue as in $G^r$, and there is an edge from the new initial vertex $v_{\text{restart}}$ of $G$ to the old initial vertex $v_0$ of $G^r$. We define $n$ priority functions as follows: for each $m \in \{1,\ldots,n\}$, we let $p_m(v_{\text{restart}}) = 0$, $p_m(b_{v,v'}) = 2$ for all $b_{v,v'} \in V_{\text{branch}}$, $p_m(v_0) = 0$ if $v_0 \in U_m$ and $1$ otherwise, and for all the other vertices $v$ of $V$, we let $p_m(v) = 0$ if $v \in U_m$ and $2$ otherwise. Notice that each vertex has only even priorities except $v_0$ which has odd priority $p_m(v_0) = 1$ whenever $v_0 \not\in U_m$.

We claim that $\playerOne$ has a winning strategy for objective $\Obj = \cap_{m=1}^n \FWP(\lambda,p_m)$ in $G$ with $\lambda = 2 \cdot n \cdot |V^r|$ if and only if he has a winning strategy for the generalized reachability objective $\Obj^r$ in $G^r$.

We first prove the left-to-right implication. Consider a winning strategy $\sigma_1$ of $\playerOne$ in $(G,\Obj)$ from $v_{\text{restart}}$. By definition $\sigma_1$ ensures victory against any strategy of $\playerTwo$. In particular, it must win against a strategy~$\sigma_2$ that plays in rounds, such that at each round, it chooses edges $(b_{v,v'}, v')$ during the first $(\lambda-1)$ steps of the round, then chooses edge $(b_{v,v'}, v_{\text{restart}})$ and starts a new round. Consequently, strategy $\sigma_1$ must eventually be able to close all windows in $(\lambda-1)$ steps without using the priorities $0$ from $v_{\text{restart}}$, otherwise $\playerTwo$ can create infinitely-many $\lambda$-bad windows and win the game. The only way to achieve this is to visit all target sets $U_m$ within $(\lambda-1)$ steps. Consider the point from which $\sigma_1$ closes all $\lambda$-windows (it exists since $\sigma_1$ wins for~$\Obj$). Observe that from this point on, this strategy can be mimicked in $G^r$ as during the first $(\lambda-1)$ steps of the round, state $v_{\text{restart}}$ is not visited, hence there is a bijection between histories in both games. By construction, the mimicking strategy $\sigma_1^r$ ensures that all target sets are visited in at most $(n \cdot |V^r| - 1)$ steps (as we get rid of the additional vertices of $V$). Hence it is winning for $\Obj^r$.

Second, we prove the right-to-left implication. Consider now a winning strategy $\sigma_1^r$ of $\playerOne$ in $(G^r,\Obj^r)$ from~$v_0$ that ensures reaching all sets $U_m$ in strictly less than $n \cdot |V^r|$ steps (w.l.o.g.). We define the strategy~$\sigma_1$ of $\playerOne$ in $(G,\Obj)$ that mimics $\sigma_1^r$ and each time $v_{\text{restart}}$ is reached restarts mimicking $\sigma_1^r$ from $v_0$. Let $\pi$ be a play consistent with $\sigma_1$. For all dimension $m \in \{1,\ldots,n\}$, for all position $j \geq 0$, we must prove that there exists $l \in \{0, 1, \ldots{}, \lambda-1\}$ such that $\pi[j+l]$ is the \smallest\ priority in $\pi[j, j+l]$. Since we only have priorities in $\{0, 1, 2\}$ and $v_0$ is the only vertex which can have priority $1$, this boils down to checking that each position~$j$ such that $\pi[j] = v_0$ is followed within $(\lambda -1)$ steps by a position $(j+l)$ such that $p_m(\pi[j+l]) = 0$.
This is obviously the case if there exists $l \leq \lambda -1$ such that $\pi[j+l] = v_{\text{restart}}$. Otherwise, as $\sigma_1$ mimics $\sigma_1^r$ and each edge of $E^r$ has been split in $G$ into two edges, $\pi[j] = v_0$ must be followed by a vertex $v \in U_m$ in strictly less than $2 \cdot n \cdot |V^r| = \lambda$ steps, by definition of $\sigma_1^r$. Hence, there exists $l < \lambda$ such that $p_m(\pi[j+l]) = 0$, which proves that $\pi \in \Obj$ and establishes the right-to-left implication.

Finally, we discuss why this reduction also holds for the other objectives mentioned in Lemma~\ref{lem:memexp}. First, the same proof holds for the direct variant $\DirFWP$: the left-to-right implication is obvious since this objective is included in $\FWP$ (Proposition~\ref{prop:inclusions}, Item~\ref{item:direct}), and the right-to-left implication actually yields a strategy $\sigma_1$ which never produces any $\lambda$-bad window, so it also wins for $\DirFWP$. Second, the proof also holds for the bounded variants $\BWP$ and $\DirBWP$. Indeed, for the left-to-right implication, the fact that strategy $\sigma_1$ must be able to close all windows (hence reach all target sets) suffice (whatever the number of steps taken), whereas the right-to-left implication is obvious since the fixed variants are included in the bounded ones by Proposition~\ref{prop:inclusions}, Item~\ref{item:bounded}. Third, all corresponding \textsf{PR} variants can be obtained through the equality of the bounded variants for \textsf{PR} and \textsf{WP} (Proposition~\ref{prop:inclusions}, Item~\ref{item:samebnd}), and the equality of the fixed variants for \textsf{PR} and \textsf{WP} in the particular case where $d = 2$ (Lemma~\ref{lem:sameobj}), as in this reduction.\qed
\end{proof} 

For the reader's interest, we complement Lemma~\ref{lem:memexp} with an example illustrating the need for exponential memory for $\playerOne$ in $\FWP$ games (it also works for the other objectives of Lemma~\ref{lem:memexp}).

\begin{example}
Consider the family of game structures depicted in Figure~\ref{fig:multiDimExpFamily}. This family is parameterized by $n \in \mathbb{N}_0$ and is inspired by a similar one proposed in~\cite{CRR14} for a different context (i.e., energy games). For each of these games structures, the number of vertices is linear in~$n$ ($|V| = 6n$) and we define $2n$ priority functions in the following way: for all $i \in \{1, \ldots, n\}$ and for all $m \in \{1, \ldots,  2n\}$, $p_m(v_i) = p_m(u_i) = 2$, 
\begin{eqnarray*}
 &p_m(v_{i,L}) = 
\begin{cases}
1 & \mbox{if $m = 2i-1$} \\
2 & \mbox{otherwise}
\end{cases},\ & \qquad p_m(v_{i,R}) =
\begin{cases}
1 & \mbox{if $m = 2i$} \\
2 & \mbox{otherwise}
\end{cases},\\
& p_m(u_{i,L}) = 
\begin{cases}
0 & \mbox{if $m = 2i-1$} \\
2 & \mbox{otherwise}
\end{cases},\ & \qquad p_m(u_{i,R}) = 
\begin{cases}
0 & \mbox{if $m = 2i$} \\
2 & \mbox{otherwise}
\end{cases}.\\
\end{eqnarray*}
Let $\Obj = \cap_{m=1}^{2n} \FWP(3n, p_m)$ be the objective of $\playerOne$. In order to prevent $(3n)$-bad windows, $\playerOne$ has to choose $u_{i,L}$ (resp.~$u_{i,R}$) whenever $\playerTwo$ chooses $v_{i,L}$ (resp.~$v_{i,R}$). Hence in order to prevent outcomes with infinitely-many $(3n)$-bad windows, $\playerOne$ must be able to record $2^n$ different histories from $v_1$ to $u_1$. This obviously requires exponential memory in $n$, hence in the size of the game.\hfill$\triangleleft$

\end{example}

  \vspace{-12mm}
\begin{figure}[htb]
  \centering   
  \scalebox{0.8}{\begin{tikzpicture}[->,>=stealth',shorten >=1pt,auto,node
    distance=2.5cm,bend angle=45,scale=0.4, font=\small]
    \tikzstyle{p1}=[draw,circle,text centered,minimum size=8mm]
    \tikzstyle{p2}=[draw,rectangle,text centered,minimum size=8mm]
    \node[p2]  (0)  at (0, 0) {$v_{1}$};
    \node[p2]  (1) at (4, 2) {$v_{1,L}$};
    \node[p2]  (2) at (4, -2)  {$v_{1,R}$};
    \node[p2]  (3) at (8, 0)  {$v_{n}$};
    \node[p2]  (4)  at (12, 2) {$v_{n,L}$};
    \node[p2]  (5)  at (12, -2) {$v_{n,R}$};
    \node[p1]  (6)  at (16, 0) {$u_{1}$};
    \node[p1]  (7) at (20, 2) {$u_{1,L}$};
    \node[p1]  (8) at (20, -2)  {$u_{1,R}$};
    \node[p1]  (9) at (24, 0)  {$u_{n}$};
    \node[p1]  (10)  at (28, 2) {$u_{n,L}$};
    \node[p1]  (11)  at (28, -2) {$u_{n,R}$};
    \coordinate[shift={(-5mm,0mm)}] (init) at (0.west);
    \path
    (init) edge (0);
    \draw[->,>=latex] (0) to (1);
    \draw[->,>=latex] (0) to (2);
    \draw[->,>=latex] (3) to (4);
    \draw[->,>=latex] (3) to (5);
    \draw[dotted,->,>=latex] (1) to (3);
    \draw[dotted,->,>=latex] (2) to (3);
    \draw[->,>=latex] (4) to (6);
    \draw[->,>=latex] (5) to (6);
    \draw[->,>=latex] (6) to (7);
    \draw[->,>=latex] (6) to (8);
    \draw[dotted,->,>=latex] (7) to (9);
    \draw[dotted,->,>=latex] (8) to (9);
    \draw[->,>=latex] (9) to (10);
    \draw[->,>=latex] (9) to (11);
    \draw[->,>=latex] (10) to[out=160,in=60] (0);
    \draw[->,>=latex] (11) to[out=200,in=300] (0);
      \end{tikzpicture}}
\vspace*{-12mm}
      \caption{Family of multi-dimension games requiring exponential memory for $\playerOne$ for objective $\FWP$ with $\lambda = 3n$.}
      \label{fig:multiDimExpFamily}
  \end{figure}
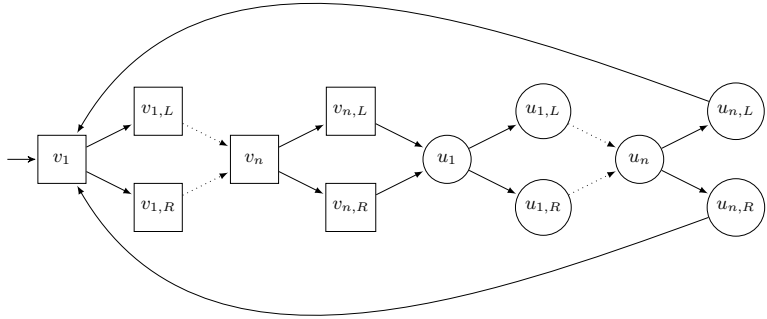
  \vspace{-5mm}

\subsection{Fixed variants} \label{subsec:multifixed}

As in the one-dimension setting, some differences arise for the fixed variants. While both \textsf{PR} and \textsf{WP} variants prove to be $\sf EXPTIME$-complete, the situation is slightly better for \textsf{WP} variants as we obtain an algorithm exponential in both the number of dimensions $n$ and the binary encoding of $\lambda$, whereas for \textsf{PR} we also get an additional exponential in the largest priority $d$ (which can be as large as the game structure). This distinction may be of importance in practical applications. We study the \textsf{PR} objective in Theorem~\ref{prop:multiPR} and the \textsf{WP} one in Theorem~\ref{prop:multiWP}.

\paragraph*{\bf Parity-response objectives.} To establish an exponential-time algorithm for multi-dimension $\DirPR$ (resp.~$\PR$) games, we reduce those games to safety (resp.~co-B\"uchi) games on an exponentially-larger game structure. Our reduction is in the same spirit\footnote{Note that the other algorithm suggested in Theorem~\ref{thm:fixed}, Item~\ref{item:fixed3} and exponential in $\lambda$ is not interesting here, since $\lambda$ can be exponential before the fixed variant becomes equivalent to the bounded one (Lemma~\ref{lem:multibornes}), hence this algorithm would take \textit{doubly}-exponential time.} as the one for Theorem~\ref{thm:fixed}, Item~\ref{item:fixed2} (i.e., case $d$ fixed in the one-dimension setting). That is, the extended structure encodes for each odd priority in each dimension, the number of steps since seeing the odd priority without seeing a \smaller\ priority in the meantime. The complexity and memory lower bounds follow from Lemma~\ref{prop:exptimehard} and Lemma~\ref{lem:memexp}.

\begin{theorem}\label{prop:multiPR} 
Let $G = (V_1,V_2,E)$ be a game structure, $v_0$ be an initial vertex, $p_1,\ldots,p_n$ be $n$ priority functions, and $\Obj$ be the objective $\cap_{m=1}^n \DirPR(\lambda,p_m)$ (resp. $\cap_{m=1}^n \PR(\lambda,p_m)$) for $\lambda \in \Nzero$. Deciding the winner in $(G,\Obj)$ from $v_0$ is $\sf EXPTIME$-complete with an algorithm in $\mathcal{O}((|V|+|E|) \cdot \lambda^{\frac{d}{2}\cdot n})$ (resp. $\mathcal{O}(|V|^2 \cdot \lambda^{d\cdot n})$) time, and exponential-memory strategies with $\mathcal{O}(\lambda^{\frac{d}{2}\cdot n})$ memory are sufficient for both players.
Moreover, exponential-memory strategies are necessary for both players.
\end{theorem}

\begin{proof}

In order to prove $\sf EXPTIME$-membership and that exponential-memory strategies are sufficient, we use the same construction as in the proof of Theorem~\ref{thm:fixed}, Item~\ref{item:fixed2}, except that we have to deal with multiple dimensions instead of one. More precisely, for the undirect variant $\PR$, we construct a new game structure~$G'$ from~$G$ such that the set $V'$ of vertices is equal to $V \times \{\bot,0,1,\ldots,\lambda - 2\}^{\frac{d}{2}\cdot n} \cup \beta_V$ with $\beta_V = \{\beta_v \mid v \in V\}$, and the set $E'$ of edges is defined as in the one-dimension case with the different dimensions considered component-wise. We define $U' = V' \setminus \beta_V$ and $\Obj' = \CoBuchi(U')$. With this construction, winning in the initial game $(G,\Obj)$ is equivalent to winning in the new game $(G',\Obj')$. For the direct variant $\DirPR$, the construction is similar except that $\beta_V$ is reduced to a unique absorbing vertex $\beta$ and $\Obj' = \Safe(U')$. In both constructions, the size of the game $G'$ is exponential in the size of the original game $G$, with $\mathcal{O}(|V'|) = \mathcal{O}(|V| \cdot \lambda^{\frac{d}{2}\cdot n})$ and $\mathcal{O}(|E'|) = \mathcal{O}(|E| \cdot \lambda^{\frac{d}{2}\cdot n})$. The stated complexities and memory requirements follow.

Finally, {\sf EXPTIME}-hardness of the decision problem follows from Lemma~\ref{prop:exptimehard} and the necessity of exponential memory for both players follows from Lemma~\ref{lem:memexp}.\qed
\end{proof}

\paragraph*{\bf Window parity objectives.} To conclude our study of \textsf{PR} and \textsf{WP} objectives, it remains to establish an exponential-time algorithm for multi-dimension $\DirFWP$ (resp.~$\FWP$) games. Again, we reduce those games to safety (resp.~co-B\"uchi) games on an exponentially-larger game structure. Our reduction is here based on the one used in the one-dimension setting (Theorem~\ref{prop:WPonedim}). That is, the extended structure encodes, for each dimension, the minimum priority of the current window, and the number of steps in that window. The complexity and memory lower bounds follow from Lemma~\ref{prop:exptimehard} and Lemma~\ref{lem:memexp}.

\begin{theorem}\label{prop:multiWP}
Let $G = (V_1,V_2,E)$ be a game structure, $v_0$ be an initial vertex, $p_1,\ldots,p_n$ be $n$ priority functions, and $\Obj$ be the objective $\cap_{m=1}^n \DirFWP(\lambda,p_m)$ (resp. $\cap_{m=1}^n \FWP(\lambda,p_m)$) for $\lambda \in \Nzero$. Deciding the winner in $(G,\Obj)$ from $v_0$ is $\sf EXPTIME$-complete with an algorithm in $\mathcal{O}((|V|+|E|)\cdot (d \cdot \lambda)^n)$ (resp. $\mathcal{O}(|V|^2 \cdot (d \cdot \lambda)^{2\cdot n})$) time, and exponential-memory strategies are sufficient for both players with $\mathcal{O}((d \cdot \lambda)^n)$ memory. Moreover, exponential-memory strategies are necessary for both players.
\end{theorem}

\begin{proof}
To establish $\sf EXPTIME$-membership and that exponential-memory strategies are sufficient, we proceed as in the one-dimension case (see the proof of Theorem~\ref{prop:WPonedim}). For the undirect variant $\FWP$, we construct from $G$ a new game structure $G'$ such that $V' = V \times (\{0,\ldots,d\} \times\{0,\ldots,\lambda - 1\})^n \cup \beta_V$ with $\beta_V = \{\beta_v \mid v \in V\}$, and the set $E'$ of edges is defined as in the one-dimension case with the different dimensions considered component-wise. Again $U' = V' \setminus \beta_V$ and $\Obj' = \CoBuchi(U')$. For the direct variant $\DirFWP$, $\beta_V$ is replaced by $\{\beta\}$ and $\Obj' = \Safe(U')$. In both cases, the size of $G'$ is exponential in the size of $G$, with $\mathcal{O}(|V'|) = \mathcal{O}(|V| \cdot (d \cdot \lambda)^n)$ and $\mathcal{O}(|E'|) = \mathcal{O}(|E| \cdot (d \cdot \lambda)^n)$. The stated complexities and memory requirements follow. To conclude, {\sf EXPTIME}-hardness follows from Lemma~\ref{prop:exptimehard} and the necessity of exponential memory follows from Lemma~\ref{lem:memexp}.
\qed\end{proof}

\paragraph*{\bf{Acknowledgments.}} We express our gratitude to Jean-Fran\c{c}ois Raskin (Universit\'e libre de Bruxelles, Belgium) and Martin Zimmermann (Saarland University, Germany) for insightful discussions.

\bibliographystyle{abbrv}
\bibliography{biblio}

\begin{thebibliography}{10}

\bibitem{DBLP:conf/csl/BaierKKW14}
C.~Baier, J.~Klein, S.~Kl{\"{u}}ppelholz, and S.~Wunderlich.
\newblock Weight monitoring with linear temporal logic: complexity and
  decidability.
\newblock In {\em Proc. of CSL-LICS}, pages 11:1--11:10. {ACM}, 2014.

\bibitem{Beeri80}
C.~Beeri.
\newblock On the membership problem for functional and multivalued dependencies
  in relational databases.
\newblock {\em {ACM} Trans. Database Syst.}, 5(3):241--259, 1980.

\bibitem{BHR16}
V.~Bruy\`ere, Q.~Hautem, and M.~Randour.
\newblock Window parity games: an alternative approach toward parity games with
  time bounds.
\newblock In {\em Proc. of GandALF}, EPTCS 226, pages 135--148, 2016.

\bibitem{BuhrkeLV96}
N.~Buhrke, H.~Lescow, and J.~V{\"{o}}ge.
\newblock Strategy construction in infinite games with {S}treett and {R}abin
  chain winning conditions.
\newblock In {\em Proc. of TACAS}, LNCS 1055, pages 207--224. Springer, 1996.

\bibitem{DBLP:journals/jacm/ChandraKS81}
A.~K. Chandra, D.~Kozen, and L.~J. Stockmeyer.
\newblock Alternation.
\newblock {\em J. ACM}, 28(1):114--133, 1981.

\bibitem{Chatterjee0RR15}
K.~Chatterjee, L.~Doyen, M.~Randour, and J.-F. Raskin.
\newblock Looking at mean-payoff and total-payoff through windows.
\newblock {\em Information and Computation}, 242:25--52, 2015.

\bibitem{ChatterjeeF13}
K.~Chatterjee and N.~Fijalkow.
\newblock Infinite-state games with finitary conditions.
\newblock In {\em Proc. of CSL}, LIPIcs 23, pages 181--196. Schloss Dagstuhl -
  LZI, 2013.

\bibitem{ChatterjeeH14}
K.~Chatterjee and M.~Henzinger.
\newblock Efficient and dynamic algorithms for alternating {B}{\"{u}}chi games
  and maximal end-component decomposition.
\newblock {\em J. {ACM}}, 61(3):15:1--15:40, 2014.

\bibitem{ChatterjeeHH09}
K.~Chatterjee, T.~A. Henzinger, and F.~Horn.
\newblock Finitary winning in $\omega$-regular games.
\newblock {\em {ACM} Trans. Comput. Log.}, 11(1), 2009.

\bibitem{ChatterjeeHH11}
K.~Chatterjee, T.~A. Henzinger, and F.~Horn.
\newblock The complexity of request-response games.
\newblock In {\em Proc. of LATA}, LNCS 6638, pages 227--237. Springer, 2011.

\bibitem{DBLP:conf/fossacs/ChatterjeeHP07}
K.~Chatterjee, T.~A. Henzinger, and N.~Piterman.
\newblock Generalized parity games.
\newblock In {\em Proc. of FOSSACS}, LNCS 4423, pages 153--167. Springer, 2007.

\bibitem{CRR14}
K.~Chatterjee, M.~Randour, and J.-F. Raskin.
\newblock Strategy synthesis for multi-dimensional quantitative objectives.
\newblock {\em Acta Informatica}, 51(3-4):129--163, 2014.

\bibitem{DziembowskiJW97}
S.~Dziembowski, M.~Jurdzinski, and I.~Walukiewicz.
\newblock How much memory is needed to win infinite games?
\newblock In {\em Proc. of LICS}, pages 99--110. {IEEE} Computer Society, 1997.

\bibitem{EmersonJ88}
E.~A. Emerson and C.~S. Jutla.
\newblock The complexity of tree automata and logics of programs (extended
  abstract).
\newblock In {\em Proc. of FOCS}, pages 328--337, 1988.

\bibitem{EmersonJ91}
E.~A. Emerson and C.~S. Jutla.
\newblock Tree automata, mu-calculus and determinacy (extended abstract).
\newblock In {\em Proc. of FOCS}, pages 368--377. {IEEE} Computer Society,
  1991.

\bibitem{DBLP:conf/cav/EmersonJS93}
E.~A. Emerson, C.~S. Jutla, and A.~P. Sistla.
\newblock On model-checking for fragments of $\mu$-calculus.
\newblock In {\em Proc. of CAV}, LNCS 697, pages 385--396. Springer, 1993.

\bibitem{FijalkowH13}
N.~Fijalkow and F.~Horn.
\newblock The surprizing complexity of generalized reachability games.
\newblock {\em CoRR}, abs/1010.2420, 2010.

\bibitem{DBLP:journals/corr/abs-1207-0663}
N.~Fijalkow and M.~Zimmermann.
\newblock Parity and {S}treett games with costs.
\newblock {\em LMCS}, 10(2), 2014.

\bibitem{2001automata}
E.~Gr{\"{a}}del, W.~Thomas, and T.~Wilke, editors.
\newblock {\em Automata, Logics, and Infinite Games: {A} Guide to Current
  Research}, LNCS 2500. Springer, 2002.

\bibitem{Hor05-GDV}
F.~Horn.
\newblock Streett games on finite graphs.
\newblock In {GDV}'05, 2005.

\bibitem{Immerman81}
N.~Immerman.
\newblock Number of quantifiers is better than number of tape cells.
\newblock {\em J. Comput. Syst. Sci.}, 22(3):384--406, 1981.

\bibitem{Jurdzinski98}
M.~Jurdzinski.
\newblock Deciding the winner in parity games is in {UP} $\cap$ co-{UP}.
\newblock {\em Inf. Process. Lett.}, 68(3):119--124, 1998.

\bibitem{Jurdzinski00}
M.~Jurdzinski.
\newblock Small progress measures for solving parity games.
\newblock In {\em Proc. of STACS}, LNCS 1770, pages 290--301. Springer, 2000.

\bibitem{DBLP:journals/siamcomp/JurdzinskiPZ08}
M.~Jurdzinski, M.~Paterson, and U.~Zwick.
\newblock A deterministic subexponential algorithm for solving parity games.
\newblock {\em {SIAM} J. Comput.}, 38(4):1519--1532, 2008.

\bibitem{DBLP:journals/fmsd/KupfermanPV09}
O.~Kupferman, N.~Piterman, and M.~Y. Vardi.
\newblock From liveness to promptness.
\newblock {\em FMSD}, 34(2):83--103, 2009.

\bibitem{Martin75}
D.~A. Martin.
\newblock Borel determinacy.
\newblock {\em Annals of Mathematics}, 102(2):363--371, 1975.

\bibitem{PitermanP06}
N.~Piterman and A.~Pnueli.
\newblock Faster solutions of {R}abin and {S}treett games.
\newblock In {\em Proc. of LICS}, pages 275--284. {IEEE} Computer Society,
  2006.

\bibitem{randourECCS}
M.~Randour.
\newblock Automated synthesis of reliable and efficient systems through game
  theory: A case study.
\newblock In {\em Proc. of ECCS 2012}, Springer Proceedings in Complexity XVII,
  pages 731--738. Springer, 2013.

\bibitem{DBLP:conf/fsttcs/Schewe07}
S.~Schewe.
\newblock Solving parity games in big steps.
\newblock In {\em Proc. of FSTTCS}, LNCS 4855, pages 449--460. Springer, 2007.

\bibitem{Thomas97}
W.~Thomas.
\newblock Languages, automata, and logic.
\newblock In {\em Handbook of Formal Languages}, volume 3, Beyond Words,
  chapter~7, pages 389--455. Springer, 1997.

\bibitem{WallmeierHT03}
N.~Wallmeier, P.~H{\"{u}}tten, and W.~Thomas.
\newblock Symbolic synthesis of finite-state controllers for request-response
  specifications.
\newblock In {\em Proc. of CIAA}, LNCS 2759, pages 11--22. Springer, 2003.

\bibitem{Weinert016}
A.~Weinert and M.~Zimmermann.
\newblock Easy to win, hard to master: Optimal strategies in parity games with
  costs.
\newblock In {\em Proc. of CSL}, LIPIcs 62, pages 31:1--31:17. Schloss Dagstuhl
  - LZI, 2016.

\bibitem{Zielonka98}
W.~Zielonka.
\newblock Infinite games on finitely coloured graphs with applications to
  automata on infinite trees.
\newblock {\em Theor. Comput. Sci.}, 200(1-2):135--183, 1998.

\end{thebibliography}

\end{document}